\documentclass[11pt,times,letter]{article}
\usepackage[margin=1in]{geometry}

\usepackage{amsmath, amsfonts,amsthm,amssymb,bm}
\usepackage{dsfont}
\usepackage[english]{babel}

\usepackage[normalem]{ulem}

\usepackage{latexsym,amssymb,amsfonts,graphicx}
\usepackage{amsmath,dsfont}
\usepackage{verbatim}
\usepackage{mathrsfs}
\usepackage{bm}
\usepackage{color}
\usepackage{epsfig}
\usepackage{appendix}
\usepackage{float}

\usepackage{url}
\usepackage{bm}
\usepackage{natbib}

\usepackage[unicode=true,colorlinks,citecolor=linkcolor,linkcolor=linkcolor,urlcolor=blue]{hyperref}

\newcommand{\Px}{\mathbb{ P} }
\newcommand{\Qx}{ \mathbb{Q} }

\newcommand{\Ex}{ \mathbb{E} }

\newcommand{\gt}{\mathcal{G}}

\numberwithin{equation}{section}

\def\esssup_#1{\underset{#1}{\mathrm{ess\,sup\, }}}
\def\essinf_#1{\underset{#1}{\mathrm{ess\,inf\, }}}
\def\argmax_#1{\underset{#1}{\mathrm{arg\,max\, }}}
\def\argmin_#1{\underset{#1}{\mathrm{arg\,min\, }}}

\newcommand{\Gx}{\mathbb{G}}
\newcommand{\Fx}{\mathbb{F} }
\newcommand{\cO}{\mathcal{O}}

\newcommand{\F}{\mathcal{F}}
\newcommand{\R}{\mathbb{R}}

\newtheorem{theorem}{Theorem}[section]
\newtheorem{definition}{Definition}[section]

\newtheorem{remark}[theorem]{Remark}
\newtheorem{lemma}[theorem]{Lemma}

\allowdisplaybreaks[4]

\definecolor{linkcolor}{rgb}{0.1,0,0.7}
\definecolor{urlcolor}{rgb}{1,0,0}

\title{Mean Field Game of Optimal Relative Investment with\\ Jump Risk}
\author{Lijun Bo \thanks{Email: lijunbo@xidian.edu.cn, School of Mathematics and Statistics, Xidian University, Xi'an, 710126, China, and School of Mathematical Sciences, University of Science and Technology of China, Hefei, 230026, China.}
\and
Shihua Wang \thanks{Email: shihuawang@ustc.edu.cn, School of Mathematical Sciences, University of Science and Technology of China, Hefei, Anhui Province, 230026, China.}
\and
Xiang Yu \thanks{Email: xiang.yu@polyu.edu.hk, Department of Applied Mathematics, The Hong Kong Polytechnic University, Hung Hom, Kowloon, Hong Kong.}
}
\date{\vspace{-5ex}}

\begin{document}
\maketitle
\begin{abstract}
This paper studies the $n$-player game and the mean field game under the CRRA relative performance on terminal wealth, in which the interaction occurs by peer competition. In the model with $n$ agents, the price dynamics of underlying risky assets depend on a common noise and contagious jump risk modelled by a multi-dimensional nonlinear Hawkes process. With a continuum of agents, we formulate the MFG problem and characterize a deterministic mean field equilibrium in an analytical form under some conditions, allowing us to investigate some impacts of model parameters in the limiting model and discuss some financial implications. Moreover, based on the mean field equilibrium, we construct an approximate Nash equilibrium for the $n$-player game when $n$ is sufficiently large. The explicit order of the approximation error is also derived.

\vspace{0.4 cm}

\vspace{0.2 cm}

\end{abstract}

\section{Introduction}\label{sec:intr}

\noindent{\it The Model Setup.} In this paper, we consider a financial market model with $n$ agents. Each agent $i$ invests in a common riskless bond and one individual stock $i$. The common time horizon for all agents is denoted by $T>0$. For $i = 1, \ldots, n$, the price process of the $i$th stock follows the following SDE:
\begin{equation}\label{eq:default-stock}
  \frac{d S_t^i}{S_{t-}^i} = (r+b_i)dt +\sigma_idW_t^i + \sigma_i^0 dW_t^0 -  dM_t^{f,i},\quad t\in[0,T],
\end{equation}
with the given parameters $b_i>0$, $\sigma_i,\sigma_i^0>0$. Here, $r\geq 0$ represents the riskless interest rate, and $(W_t^0,W_t^1,\ldots,W_t^n)^{\top}_{t\in[0,T]}$ is an $n+1$-dimensional Brownian motion under the filtered probability space $(\Omega,\F,\Fx,\Px)$ with the reference filtration $\Fx=(\F_t)_{t\in[0,T]}$ satisfying the usual conditions. The Brownian motion $W^0=(W_t^0)_{t\in[0,T]}$ appears in all price dynamics, which represents a common noise in the financial market. The Brownian motion $W^i=(W_t^i)_{t\in[0,T]}$, specified to each individual risky asset, stands for the idiosyncratic noise. In addition, each stock is subject to the downward jump risk and the jump contagion among all stocks is allowed. In particular, we denote $ \bm{N} := (N_t^1,\ldots,N_t^n)^{\top}_{t\in[0,T]}$ as an $n$-dimensional mutually exciting point process modelled by a nonlinear Hawkes process with a (bounded, Lipschtiz and differentiable) jump rate function $f: \mathbb{R}_+\rightarrow \mathbb{R}_{+}$.\footnote{For example, it is referred as the jump function in \cite{Chevallier2017} in the context of generalized Hawkes processes; and it is called the spiking rate function by \cite{Locherbach18} in the context of interacting neurons.} The intensity process is defined by $\bm{\Lambda}^{f}:= (\Lambda_t^{f,1},\ldots,\Lambda_t^{f,n})_{t\in[0,T]}$, and the compensated process of $\bm{N}$, defined by $\bm{M}_t^f := \bm{N}_t-\int_0^t\bm{\Lambda}_s^{f} ds$, $t\in[0,T]$, is an $n$-dimensional $(\Px,\Gx)$-martingale, where $\bm{M}_t^f=(M_t^{f,1},\ldots,M_t^{f,n})^{\top}$. Namely, for each $i=1,\ldots,n$, $M_t^{f,i}:=N_t^{i}-\int_0^t\Lambda_s^{f,i}ds$ for $t\in[0,T]$ is a scalar $(\Px,\Gx)$-martingale.

The global market filtration $\Gx=(\gt_t)_{t\in[0,T]}$ is defined by $\gt_t:=\F_t\vee\sigma({\bm N}_s;~s\leq t)$ as the right-continuous augmentation by null sets (see \cite{GTM-113}, Definition 7.2 in Chapter 2). By \cite{bocapponisicon2018}, the Brownian motion under $\Fx$ is also a Brownian motion under $\Gx$, i.e., the immersion property holds. It is assumed in the present paper that the vector intensity process $\bm{\Lambda}^{f}=(\Lambda_t^{f,1},\ldots,\Lambda_t^{f,n})_{t\in[0,T]}$ is governed by
\begin{equation}\label{eq:lambda-i}
  \Lambda_t^{f,i} = f(\lambda_t^i), \ \ \   d\lambda_t^i = \alpha_i(\lambda_{\infty}^i - \lambda_t^i)dt + \frac{\beta_i }{n}\sum_{j=1}^{n}\varsigma_jdN^j_t,\quad \lambda_0^i>0,\quad i=1,\ldots,n,
\end{equation}
where $\lambda_{\infty}^i>0$ is the mean-reverting level of the underlying intensity factor of stock $i$ with speed $\alpha_i>0$, $\beta_i>0$ describes the scaled jump size effect to the intensity factor of stock $i$, and $\varsigma_{j}>0$ measures the contagion effect on the intensity factor of stock $i$ by the jump of stock $j$. The contagious risk is then captured because the downward jump of one stock increases the jump intensity of all other stocks, leading to a higher risk of default clustering (see \cite{bocapponichen2019}).

Let the $\Gx$-predictable process $\pi_t^i$ be the proportion of wealth that the agent $i$ allocates in the stock $S^i$ at time $t$. The self-financing wealth process of agent $i$ under the control $\pi^i=(\pi_t^i)_{t\in[0,T]}$ is given by
\begin{align}\label{eq:wealth-dynamics-i}
  X_t^i= x_i+\int_0^t \left(rX_s^i + b_i\pi_s^iX_s^i\right)ds + \int_0^t\pi_s^iX_s^i\sigma_idW_s^i  + \int_0^t\pi_s^iX_s^i\sigma_i^0dW_s^0 - \int_0^t\pi_s^iX_{s-}^idM_s^{f,i},
\end{align}
where $x_i>0$ denotes the initial wealth of agent $i$. The portfolio vector is denoted by ${\bm\pi}:= (\pi_t^1,\ldots, \pi_t^n)_{t\in[0,T]}$. Let us denote ${\cal A}^i$ the set of admissible controls for the agent $i$. We say a control process $\pi^i=(\pi_t^i)_{t\in[0,T]}\in {\cal A}^i$ is admissible if $\pi^i$ is $\Gx$-predictable and satisfies $D_0\leq\pi_t^i\leq 1-\epsilon_0$ for some constant $D_0 \in \mathbb{R}$ and positive constant $\epsilon_0 \ll 1$ (both $D_0$ and $\epsilon_0$ depend on the control) such that the non-bankruptcy condition $X_t^i> 0$ holds a.s. for $t\in[0,T]$. Note that the pure jump martingale $M^{f,i}=(M_t^{f,i})_{t\in[0,T]}$ has the jump size of one. In view of \eqref{eq:wealth-dynamics-i}, the wealth process $X^i=(X_t^i)_{t\in[0,T]}$ must be positive a.s. if the initial wealth level $x_i>0$ because the admissible control is constrained to satisfy $\pi_t^i<1$ for all $t\in[0,T]$.

Each agent in the market aims to maximize the expected utility with a competition component represented by the geometric average of the terminal wealth $X_T=(X_T^1,\ldots,X_T^n)$ from all peers. The objective function of the $i$th agent is given by
\begin{align}\label{eq:objective}
  J_i(\pi^1,\ldots \pi^n) := \Ex\left[ U_i\left(X_T^i, \bar{X}_T\right)\right] ,\quad \bar{X}_T := \left(\prod_{i=1}^n X_T^i\right)^{\frac{1}{n}},
\end{align}
in which the utility function $U_i : \R_{+}^2\rightarrow \R$ of the $i$th agent is of the CRRA type that
\begin{equation}\label{eq:utility-U-i}
  U_i(x,y):= U(xy^{-\theta_i};\gamma_i),\quad \forall (x,y)\in\R_+^2,
\end{equation}
where $U(x;\gamma)$ is a power utility that $U(x;\gamma):= \frac{1}{\gamma}x^{\gamma}$. It is assumed in the present paper that all risk aversion parameters $\gamma_i \in(0,1)$ and all competition weight parameters $\theta_i\in[0,1]$.
\begin{remark}\label{rem:thetai}
The above relative investment preference is motivated by the fact that peer competition sometimes has notable impacts on fund manager's decision making. The parameter $\theta_i$ of the agent $i$ represents how competitive the agent is towards the relative performance with her peers. The small (resp. large) value of $\theta_i$ implies a low (resp. high) relative performance concern. In the extreme case, the utility with $\theta_i=0$ reduces to the standard utility on her own absolute wealth; while the utility with $\theta_i=1$ indicates that the agent is extremely sensitive to her relative performance with other peers and no absolute performance is concerned.
\end{remark}

Due to the presence of $ \bar{X}_T$ inside the utility, the optimal decision of the agent $i$ is coupled with optimal controls by other peers, which makes the game problem challenging especially when there are jump risk contagion and common noise. In the present paper, we aim to provide an approximate Nash equilibrium to the $n$-player game problem, which is defined in the following sense.
\begin{definition}[Approximate Nash Equilibrium (ANE)]\label{def:nashequilibrium}
Let the objective functional $J_i$ be defined in \eqref{eq:objective}. An admissible strategy $\bm{\pi}^*=(\pi^{*,1},\ldots, \pi^{*,n})\in{\cal A}:=\prod_{i=1}^n{\cal A}^i$ is called a Nash equilibrium if, for all $\pi^i \in{\cal A}^i$ with $i = 1,\ldots,n$, it holds that
\begin{equation}\label{eq:def-Nash}
  J_i(\bm{\pi}^{\ast}) \geq  J_i(\pi^i,\pi^{\ast,-i}),\qquad   \text{with} \ \ \pi^{\ast,-i}:=(\pi^{*,1},\ldots,\pi^{*,i-1},\pi^{*,i+1},\ldots, \pi^{*,n}).
\end{equation}
If there exists a constant $\varepsilon_n>0$ satisfying $\lim_{n\to\infty}\varepsilon_n=0$, and it holds that
\begin{equation}\label{eq:def-epsilon-Nash}
   \sup_{\pi^i\in{\cal A}^i} J_i(\pi^i,\pi^{\ast,-i}) \leq J_i(\bm{\pi}^{\ast})+ \varepsilon_n,
\end{equation}
we call ${\bm\pi}^{\ast}$ an $\varepsilon_n$-approximate Nash equilibrium.
\end{definition}

To this end, we will first take advantage of the simplified structure of the mean field game (MFG) when $n\rightarrow\infty$, in which the impact of an individual agent on the aggregated wealth of the whole population becomes negligible. That is, comparing with $\bar{X}_T$ in \eqref{eq:objective} for $n$ agents, we now consider $m\in{\cal D}$ (with ${\cal D}$ being the set of $\Fx^0=(\F_t^0)_{t\in[0,T]}=(\sigma(W_s^0;~s\leq t))_{t\in[0,T]}$-adapted processes that are right continuous with left limits) as the geometric average wealth of the continuum of agents, which is the competition component in the utility of a representative agent. The MFG problem is to find a pair $(\pi^*, m^*)$ that solves the utility maximization problem for a representative agent similar to \eqref{eq:objective} that
\begin{align*}
\sup_{\pi\in {\cal A}_{\rm MF} } \Ex\left[U(X_T, m^*_T)\right]=\Ex\left[U(X^*_T, m^*_T)\right],
\end{align*}
where $U(\cdot)$ (resp. $X=(X_t)_{t\in[0,T]}$) is the utility function (resp. the wealth process under an arbitrary control $\pi\in{\cal A}_{\rm MF}$) of the representative agent, and $(X_t^*)_{t\in[0,T]}$ is the wealth process under the control $\pi^*$. In addition, the geometric average wealth of the population coincides with the geometric mean of the wealth process of the representative agent that $m^*_t =\exp\{\Ex[\log(X_t^{\ast })|\F_t^{0}]\}$. The precise formulation of the MFG problem and the definition of mean field equilibrium $\pi^*$ are given in Section \ref{sec:limiting}. Using the stochastic maximum principle, we are able to characterize one deterministic mean field equilibrium in analytical form. Based on the information and the structure of this mean field equilibrium, we can then construct the $\varepsilon_n$-Nash equilibrium for the $n$-player game as described in Definition~\ref{def:nashequilibrium}.\\

\noindent{\it Literature Review.}\quad Optimal investment under relative performance for a finite number of agents and a continuum of agents has been an important research topic in recent years. In a Black-Scholes model, \cite{espinosa2015optimal} and \cite{Bie2017} formulate and study the $n$-player game under the CARA relative performance using the coupled quadratic BSDE system when the equilibrium pricing and portfolio constraints are also incorporated. In a log-normal market model with deterministic parameters and common shock together the asset specialization to each agent, \cite{lacker2017mean} consider the $n$-player game and MFG problems under both CARA and CRRA relative performance on terminal wealth. Thanks to the simplified structure of the asset specialization, the constant equilibrium is obtained therein for both the $n$-player game and MFG problems. In the same framework, \cite{lacker2020many} generalize the problems to examine optimal relative performance on consumption under CRRA utilities. \cite{FuSZ} consider the generalization of the market model in \cite{lacker2017mean} by allowing random return and volatility parameters and solve the $n$-player game and MFG problems using the FBSDE approach when the exponential relative performance on terminal wealth is concerned. \cite{ReisPla2020} study the MFG under forward relative performance utilities of CARA type. \cite{KraftMS2020} formulate and solve some $n$-player games in a general stochastic volatility price model (Heston and Chacko-Viceira stochastic-volatility models as special cases) with unhedgeable stochastic factors. Under the exponential relative performance on terminal wealth, \cite{HuZ21} recently investigate the $n$-player game and MFG in the incomplete It\^{o}-diffusion market model and also in the case with random risk tolerance coefficients.

To the best of our knowledge, the $n$-player games and MFGs under relative performance when the underlying price dynamics exhibit jumps have not been studied before. On the other hand, the importance of considering defaultable risky assets, especially after the systemic failure caused by the global financial crisis, has attracted a lot of attention; see, for example,  \cite{belanger2004general}, \cite{yufan2007} and references therein. To better understand the impact of systemic default risk on dynamic portfolio management, abundant recent works have considered optimal investment problems when jumps of risky assets are contagious. See, for example, \cite{bocapponisicon2018}, \cite{bocapponichen2019}, \cite{boliaoyu2019}, \cite{shen2020mean}, \cite{boliaoyu2021} among others that are based on the interacting intensity framework, allowing the credit default in one risky asset to increase the default intensities of other surviving names. See also \cite{Yu21} in the context of optimal dividend control for an insurance group. The present paper aims to enrich the study of the $n$-player games and MFGs under relative performance by featuring the jump risk. In particular, the jump risk in price dynamics is modelled by an $n$-dimensional mutually exciting Hawkes process, whose componentwise intensity process satisfies the specific form of \eqref{eq:lambda-i}. As a result, the contagion phenomenon can be depicted in the model with $n$ agent because the jump of one stock leads to a larger jump intensity of all other stocks. Meanwhile, we adopt the asset specialization framework in \cite{lacker2017mean} with a common noise and focus on the CRRA relative performance utility. The presence of contagious jumps in price dynamics give rise to the controlled jump component.

Starting from the seminal works by \cite{lasry2007mean} and \cite{huang2006large}, MFGs have been actively studied and widely applied in economics, finance and engineering. Giving a full list of references is beyond the scope of this paper. For some recent developments in theories and applications, we refer to \cite{gueantLarsryLions11}, \cite{Bensoussan2013},   \cite{carmona2016lectures}, \cite{carmona2018probabilistic} and references therein. However, we also note the majority of existing research has focused on models when the controlled state processes have continuous paths, and the study of MFGs with controlled jumps is relatively rare. In the simple setting of inhomogeneous Poisson process, Nutz and Zhang \cite{ZN19} consider the rank-based mean field competition when each agent controls the intensity of the Poisson project process. \cite{YZZ21} further extend the model to some two-layer mean field competitions based on teamwork formulations, in which team members collaborate to control the intensity of the Poisson project process. \cite{GomesMS} and  \cite{Neumann2020} examine some MFGs with continuous time Markov chains. \cite{Hafa2014} considers the McKean-Vlasov stochastic control problems. Recently, \cite{benazzoli2020} study MFGs with controlled jump-diffusion processes, in which the jump component is driven by a Poisson process with a time-dependent intensity function. The existence of a Nash equilibrium is obtained therein by using relaxed control and martingale problem arguments. Building upon results in \cite{benazzoli2020}, \cite{benazzoli2019} further verify that the mean field equilibria can be used to construct an approximate Nash equilibrium in the $n$-player game when $n$ is large enough.\\

\noindent{\it Our Contributions.}\quad Although our targeted MFG is in the realm of MFGs with controlled jump-diffusion processes, our model and methodology differ from the ones in \cite{benazzoli2020} (see also \cite{benazzoli2019}). To be precise, the MFG problem in \cite{benazzoli2020} stems from a symmetric nonzero-sum $n$-player game, in which the Poisson jump process for each player has the same deterministic intensity and all players have the same objective function. In contrast, our $n$-player game is formulated for heterogeneous agents with different underlying processes and relative performance utilities. The contagious jump risk is a new feature of our $n$-player game and a common noise appears in all risky assets, which are not concerned by \cite{benazzoli2020}. Our mathematical contribution is two-fold. First, we model the contagious jump risk in the $n$-player game by a mutually exciting Hawkes process, which enables us to formulate a tractable MFG problem with controlled jumps under the assumption of constant type vector $o$ (see \eqref{constypevec} in the assumption $\bm{(A_O)}$ below). The strong control approach can be applied and we can characterize a deterministic mean field equilibrium in an analytical form by using the stochastic maximum principle argument; see \autoref{thm:optimal-MFG}. Some quantitative properties of the obtained mean field equilibrium are examined, yielding some interesting financial implications. Second, despite the lost tractability in the $n$-player game, we can use the mean field equilibrium to construct an $\varepsilon_n$-approximate Nash equilibrium in the model with finite $n$ agents when $n$ is sufficiently large. We highlight that the explicit convergence rate of the approximation error $\varepsilon_n$ is also obtained; see \autoref{thm:epsilon-Nash-equilibrium}. The constructed approximate Nash equilibrium using the mean field equilibrium can efficiently help to reduce the dimensionality of the $n$-player game in practical applications.

The rest of the paper is organized as follows. In Section \ref{sec:limiting}, we formulate the MFG problem in the limiting model and obtain a time-dependent deterministic mean field equilibria. Section \ref{sec:numerical} presents some quantitative properties and numerical sensitivity results on the mean field equilibrium. Section \ref{sec:epsilon-Nash-equilibrium} establishes an approximate Nash equilibrium for an $n$-player game problem. Some conclusion remarks and future directions are given in Section \ref{sec:con}. Finally, the proofs of some auxiliary results and arguments to derive the mean field limit of the default intensity are reported in \ref{sec:proofs} and \ref{sec:heu}, respectively.

\section{Mean Field Game Problem}\label{sec:limiting}

To avoid the complexity of the coupled controls in the $n$-player game problem \eqref{eq:objective}, one can consider the limiting model that enjoys the decentralized structure and the impact of an individual agent on the whole population is negligible. That is, we can first solve a stochastic control problem for a representative agent against a fixed environment (assuming that the geometric average wealth of the population is $m=(m_t)_{t\in[0,T]} \in {\cal D}$) and derive its best response portfolio strategy $\pi^{*,m}$ (as a functional of $m$). We then apply this best response strategy to generate the wealth process $X^{\ast,m}$ of the representative agent. Finally, as all agents should behave the same in the mean field model, the geometric mean of the wealth process $e^{\mathbb{E}[\log(X_t^{\ast,m})|\mathcal{F}_t^0]}$ (as a functional of $m$) of the representative agent should coincide with the geometric average wealth $m_t$ of the whole population, which gives the consistency condition to characterize a mean field equilibrium. The mathematical formulation of the MFG problem and the precise definition of a mean field equilibrium are given as follows.

For $i = 1, \ldots, n$, let us denote the type vector $o^i := (x_i,\lambda_0^i,\alpha_i,\lambda^i_\infty, \beta_i, \varsigma_i,b_i,\sigma_i, \sigma_i^0 ,\gamma_i,\theta_i)\in \mathcal{O}:= \R_{+}^9 \times(0,1)\times[0,1]$ and the space $E := \mathcal{O}\times \R_{+}^2$. Let ${\cal B}(E)$ (resp. ${\cal P}({\cal O})$) be the Borel $\sigma$-algebra generated by the open sets of $E$ (resp. the set of probability measures on ${\cal O}$).

For mathematical tractability, the following assumption is imposed throughout the paper.
\begin{itemize}
\item[] $\bm{(A_O)}$: There exists a constant vector
\begin{align}\label{constypevec}
o: = (x_0,\lambda_0,\alpha,\lambda_\infty, \beta, \varsigma,b,\sigma, \sigma^0,\gamma,\theta)\in {\cO}
\end{align}
denoting the type vector of the limiting model 
such that
\begin{align}\label{eq:constanttypevector}
   \nu_0^n:= \frac{1}{n}\sum_{i=1}^{n}\delta_{o^i} \Rightarrow \nu_0 := \delta_{o},\quad \text{as} \ n \rightarrow \infty,
\end{align}
where $\nu_0$ denotes the Dirac measure on the constant vector $o$ and the convergence holds with the order of $O(n^{-1})$ and ``$\Rightarrow$" denotes the weak convergence, i.e., $\int_{{\cal O}} fd\nu_0^n \rightarrow \int_{{\cal O}}fd\nu_0$ as $n\rightarrow \infty$ for every bounded continuous function $f$ on ${\cal O}$. In addition, it is assumed that $(b_i,\sigma_i, \sigma^0_i,\gamma_i,\theta_i)\rightarrow o_1 $ as $i \rightarrow \infty $, where $o_1: = (b,\sigma, \sigma^0,\gamma,\theta )\in {\cO_1}:= \R_{+}^3\times(0,1)\times[0,1]$.
\end{itemize}

Under the assumption $\bm{(A_O)}$ that the type vector $o\in\mathcal{O}$ is a constant vector in the mean field model, the arguments provided in \ref{sec:heu} yield that the mean field limit of the intensity process is a continuous deterministic function, which satisfies
\begin{align}\label{eq:lambda-MF}
  \lambda_t^{f,o} = f(\lambda_t^{l}),\quad \lambda_t^{l}  = \lambda_0 + \int_{0}^{t}\alpha(\lambda_{\infty} -\lambda_s^{l})ds + \int_{0}^{t} \beta \varsigma f(\lambda_s^l)  ds.
\end{align}
In the mean field model when $n\rightarrow\infty$, the wealth process of a representative agent is governed by
\begin{equation}\label{eq:state-MF}
 dX_t^o = (rX_t^o + b\pi_tX_t^o)dt + \pi_tX_{t-}^o(\sigma dW_t + \sigma^0 dW_t^0 - dM_t^f),\quad X_0^o = x_0>0,
\end{equation}
where $W=(W_t)_{t\in[0,T]}$ is a scalar Brownian motion under $(\Omega,\F,\Px)$ that is independent of Brownian motion $(W^0,W^1,\ldots,W^n)$. The pure jump martingale $M^f=(M_t^f)_{t\in[0,T]}$ satisfies the decomposition that
\begin{align*}
M_t^{f} = N_t^o - \int_{0}^{t} \lambda_s^{f,o} ds,\quad t\in[0,T],
\end{align*}
where $N^o=(N_t^o)_{t\in[0,T]}$ is a Poisson point process with the deterministic intensity process $(\lambda_t^{f,o})_{t\in[0,T]}$.

For $n$ sufficiently large, in view of the common noise $W^0$ in the wealth process $X_t^i$, we may approximate the geometric mean $\bar{X}=(\bar{X}_t)_{t\in[0,T]}$  by a c\`{a}dl\`{a}g $\Fx^0$-adapted process $m=(m_t)_{t\in[0,T]}$ (i.e., $m\in{\cal D}$). Let $\Gx^{{\rm MF}} = (\mathcal{G}_t^{\rm MF})_{t\in[0,T]}$ denote the smallest filtration satisfying the usual conditions, in which $W$, $W^0$, $N^o$ are adapted.  We denote ${\cal A}_{\rm MF}$ the set of admissible controls when $\pi=(\pi_t)_{t\in[0,T]}$ is $\Gx^{\rm MF}$-predictable and satisfies $D_0\leq \pi_t\leq 1-\epsilon_0$ with some constant $D_0\in\R$ and positive constant $\epsilon_0 \ll 1$ (depending on the control) such that $X^o=(X_t^o)_{t\in[0,T]}$ has no bankruptcy.

Let us first give the definition of mean field equilibrium.
\begin{definition}[Mean Field Equilibrium (MFE)]\label{defMFE}
For a given $\Fx^0$-adapted process $m=(m_t)_{t\in[0,T]}\in{\cal D}$, let $\pi^{\ast,m} \in \mathcal{A}_{\rm MF}$ be the best response solution to the stochastic control problem \eqref{eq:objective-limit}. The strategy $\pi^*:=\pi^{\ast,m^*}$ is called a mean field equilibrium~(MFE) if it is the best response to itself such that $m^*_t =\exp\{\Ex[\log(X_t^{\ast,m^*})|\F_t^{0}]\}$, $t\in[0,T]$, where $(X_t^{\ast,m^*})_{t\in[0,T]}$ is the wealth process in \eqref{eq:state-MF} under the control $\pi^{\ast,m^*}$. Moreover, if $\pi^*$ is deterministic, we call $\pi^*$ a deterministic MFE.
\end{definition}

Based on the above definition, finding a MFE of the mean field game problem is to solve the two-step problem:

\emph{Step 1.} For a fixed $\Fx^0$-adapted process $m=(m_t)_{t\in[0,T]} \in {\cal D}$, we solve a stochastic control problem for a representative agent against the fixed environment that
\begin{equation}\label{eq:objective-limit}
 \sup_{\pi \in {\cal A}_{\rm MF} }J(\pi)
 = \sup_{\pi\in {\cal A}_{\rm MF} } \Ex\left[U(X_T^o m_T^{-\theta};\gamma)\right]
 =  \sup_{\pi\in {\cal A}_{\rm MF} } \Ex\left[ \frac{1}{\gamma}(X_T^o)^{\gamma}m_T^{-\theta\gamma} \right],
\end{equation}
where the wealth process $X_T^o$ satisfies \eqref{eq:state-MF} with the risk aversion parameter $\gamma\in(0,1)$ and the competition parameter $\theta\in[0,1]$. The best response strategy for the representative agent is denoted by $\pi^{\ast,m} \in {\cal A}_{\rm MF}$, and $X^{\ast,m}$ stands for the wealth process in \eqref{eq:state-MF} under the control $\pi^{\ast,m}$ .

\emph{Step 2.} We next derive a MFE by the consistency condition, which is to find the fixed point $m^*$ to the equation $m_t =\exp\{\Ex[\log(X_t^{\ast, m })|\F_t^{0}]\}$ for all $t\in[0,T]$. The MFE is then given by $\pi^{\ast,m^*}$.

To facilitate the proof of the main theorem in this section, we first present the next auxiliary lemma, whose proof is reported in \ref{sec:proofs}.
\begin{lemma}\label{lem:Phi}
Define the function $ \Phi(\pi,\lambda): (-\infty, 1)\times \R_+ \rightarrow \R $ with $\sigma + \sigma^0 > 0$ and $b>0$ that
\begin{equation}\label{eq:Phi}
  \Phi(\pi,\lambda):=
  (\gamma-1)[\sigma^2 + (\sigma^0)^2] \pi -\theta\gamma(\sigma^0)^2\pi -\lambda(1-\pi)^{\gamma-1} +\lambda + b.
\end{equation}
Then, for each fixed $\lambda\in \R_+$, there exists a unique $\pi^* \in (-\infty,1)$ such that
\begin{equation}\label{eq:Phi-s}
   \Phi(\pi^*,\lambda)=0.
\end{equation}
Moreover, there exists an $\epsilon_0\in(0,1)$ small enough such that $\pi^*\in(0,1-\epsilon_0]$. Equivalently, for $\pi^*$ satisfying \eqref{eq:Phi-s}, there exists a unique continuous and decreasing function $\phi:\R_+\to (0,1-\epsilon_0]$ such that
\begin{equation}\label{eq:pi-ast-lem}
  \pi^* = \phi(\lambda),
\end{equation}
where $\phi(\lambda)$ has a continuous partial derivative with respect to $\lambda$.
\end{lemma}

We can now establish the main result of this section that gives a time-dependent deterministic mean field equilibrium in the MFG problem as a function of the deterministic limiting intensity process.
\begin{theorem}\label{thm:optimal-MFG}
There exists one deterministic MFE strategy $\pi^{\ast}=(\pi^*_t)_{t\in[0,T]}\in{\cal A}_{\rm MF}$ that satisfies
\begin{equation}\label{eq:pi-mean-limit-thm}
   \Phi(\pi_t^{\ast}, \lambda_t^{f,o}) =  0,\quad t\in[0,T].
\end{equation}
Equivalently, this deterministic MFE strategy $\pi^*=(\pi_t^{\ast})_{t\in[0,T]}$ can be written by
\begin{equation}\label{eq:pi-ast-thm}
  \pi_t^* = \phi(\lambda_t^{f,o}),\quad t\in[0,T].
\end{equation}
Here, $\Phi(\pi,\lambda)$ and $\phi(\lambda)$ are given in Lemma \ref{lem:Phi}. In addition, let $X^{\ast}=(X_t^{\ast})_{t\in[0,T]}$ be the wealth process under the deterministic MFE $\pi^*$. The associated fixed point $m^*=(m_t^*)_{t\in[0,T]}$ satisfying the consistency condition $m^*_t =\exp\{\Ex[\log(X_t^{\ast})|\F_t^{0}]\}$ for $t\in[0,T]$ is characterized by
\begin{align}\label{eq:geometric-m-thm}
  m^*_t &= x_0 \exp\left\{\int_{0}^{t}\left(\eta(s;\pi_s^*) -\frac{1}{2}(\sigma^0\pi_s^{*})^2 \right)ds + \int_{0}^{t}\sigma^0\pi_s^{*} dW_s^0\right\},
\end{align}
where the function $\eta(t;\pi)$ is given by, for $(t,\pi)\in[0,T]\times(-\infty, 1)$,
\begin{equation}\label{eq:eta-new}
  \eta(t; \pi) = r+ (b+\lambda_t^{f,o})\pi - \frac{1}{2}\sigma^2\pi^2 + \lambda_t^{f,o} \log(1-\pi).
\end{equation}
Here, $\lambda^{f,o}=(\lambda_t^{f,o})_{t\in[0,T]}$ is given by \eqref{eq:lambda-MF}.
\end{theorem}

\begin{remark}
Note that if there is no contagious jump risk in price dynamics, i.e., $\lambda_t^{f,o}\equiv 0$ (i.e., the jump rate function $f\equiv0$), the function $\Phi(\pi,0)$ defined in \eqref{eq:Phi} is reduced to
\begin{align*}
\Phi(\pi,0)= (\gamma-1)[ \sigma^2 + (\sigma^0)^2] \pi -\theta\gamma(\sigma^0)^2\pi + b,\quad \forall \pi\in\R.
\end{align*}
It follows that
\begin{align}\label{eq:pitDZ}
\pi_t^*=\phi(0)=\frac{b}{(1-\gamma)[\sigma^2+(\sigma^0)^2]+\theta\gamma(\sigma^0)^2},
\end{align}
which is a constant mean field Nash equilibrium that coincides with the result in Theorem $3.6$ of \cite{lacker2017mean}. Therefore, the obtained MFE given in \eqref{eq:pi-ast-thm} is a generalization of the constant MFE in \cite{lacker2017mean} to incorporate the jump risk.
\end{remark}

\begin{proof} [Proof of \autoref{thm:optimal-MFG}]
First, for a given process $m=(m_t)_{t\in[0,T]}\in \mathcal{D}$, we aim to solve the  mean field stochastic control problem via the stochastic maximum principle; see \cite{oksendal2005applied}. To this end, we first note that the Hamiltonian function corresponding to the control problem \eqref{eq:objective-limit} is given by
\begin{align}\label{eq:hamiltion-H-limit}
 H(t,x,\pi,p,q,q^0,y)
 & := ( rx + bx\pi) p  + \sigma\pi xq + \sigma^0\pi xq^0 - \pi x \lambda_t^{f,o} y,
\end{align}
for $(t,x,\pi,p,q,q^0,y)\in[0,T]\times\mathbb{R}_+\times U\times \mathbb{R}^4$ with the policy space $U:= (-\infty,1-\epsilon_0]$.

Let $\pi^{m}=(\pi^m_t)_{t\in[0,T]}\in{\cal A}_{\rm MF}$ be an arbitrary admissible strategy that may depend on $m$, and $X^{m}=(X_t^m)_{t\in[0,T]}$ be the corresponding wealth process under $\pi^m$. Then, the adjoint forward-backward SDEs (corresponding to $(\pi^{m},X^{m})$) are given by
\begin{equation*}
\left\{
  \begin{aligned}
  & dX_t^{m} =  \partial_pH(t,X_t^{m},\pi_t^{m},P_t^{m},Q_t^{m},Q_t^{0,m},Y_t^{m})dt + \pi_t^{m}X_t^{m}(\sigma dW_t +\sigma^0 dW_t^0 - dM_t^f), \\
  & dP_t^{m} = - \partial_x H(t,X_t^{m},\pi_t^{m},P_t^{m},Q_t^{m},Q_t^{0,m},Y_t^{m}) dt + Q_t^{m}dW_t  + Q_t^{0,m}dW_t^0 + Y_t^{m}dM_t^f,\\
  & X_0^{m}  = x_0 , \\
  &P_T^{m}  = (X_T^{m})^{\gamma-1}m_T^{-\theta\gamma},
  \end{aligned}
  \right.
\end{equation*}
where $\partial_pH$ (resp. $\partial_xH$) denotes the partial derivative of $H$ w.r.t. $p$ (resp. $x$). In view of \eqref{eq:hamiltion-H-limit},  we get that
\begin{equation}\label{eq:FBSDE-2}
\left\{
  \begin{aligned}
  & dX_t^{m} = \left(rX_t^{m} + b\pi_t^{m}X_t^{m}\right)dt  + \pi_t^{m}X_t^{m}(\sigma dW_t +\sigma^0 dW_t^0 - dM_t^f), \\
  & dP_t^{m} = -\left[rP_t^{m} +\pi_t^{m}(bP_t^{m}+\sigma Q_t^{m}+ \sigma^0Q_t^{0,m}-\lambda_t^{f,o} Y_t^{m}) \right]dt + Q_t^{m}dW_t  \\
  &\qquad \qquad + Q_t^{0,m}dW_t^0 + Y_t^{m}dM_t^f,\\
  & X_0^{m}  = x_0,  \\
  &P_T^{m}  = (X_T^{m})^{\gamma-1}m_T^{-\theta\gamma}.
  \end{aligned}
  \right.
\end{equation}
We next solve FBSDE~\eqref{eq:FBSDE-2} in terms of $(\pi^m,X^m)$ explicitly. To do this, it follows from \eqref{eq:FBSDE-2} that
\begin{align}\label{eq:log-X-ast}
  d \log X_t^{m} &= \left[r+ (b+\lambda_t^{f,o}) \pi_t^{m}- \frac{1}{2}(\sigma^2+(\sigma^0)^2)(\pi_t^{m})^2 \right]dt + \pi_t^{m}\sigma dW_t \\
  &\quad  + \pi_t^{m}\sigma^0 dW_t^0 + \log(1-\pi_{t-}^{m})dN_t^{o}. \nonumber
\end{align}
Recall that $\Fx^{0}$ is the filtration generated by the Brownian motion $W^0=(W_t^0)_{t\in[0,T]}$. Taking conditional expectations on both sides of \eqref{eq:log-X-ast} w.r.t. $\F_t^{0}$ for $t\in[0,T]$, and using Lemma~B.2 in~\cite{Giesecke2012MF}, we can deduce that
\begin{align}\label{eq:log-X-expectation}
  d \Ex[\log X_t^{m}|\F_t^0]
 & = \left\{r+(b+\lambda_t^{f,o})\Ex[\pi_t^{m}|\F_t^0] - \frac{1}{2}\left(\sigma^2+ (\sigma^0)^2\right) \Ex[(\pi_t^{m})^2|\F_t^0] \right\}dt \\
  & \quad + \sigma^0\Ex[\pi_t^{m}|\F_t^0] dW_t^0 +
     \lambda_t^o\Ex[\log\left(1-\pi_{t-}^{m}\right)|\F_t^0]dt. \nonumber
\end{align}
Let us denote $ m_t^{X}:= \exp\{\Ex[\log(X_t^{m})|\F_t^{0}]\}$ for $t\in[0,T]$. It follows from It\^{o}'s lemma that
\begin{align}\label{eq:geometric-m}
  d m_t^{X} &= d\exp\{\Ex[\log X_t^{m}| \F_t^{0}]\}\\
  & = m_t^{X}  \bigg\{r+ (b+\lambda_t^{f,o})\Ex[\pi_t^{m}|\F_t^0] - \frac{1}{2}(\sigma^2+ (\sigma^0)^2)\Ex[(\pi_t^{m})^2|\F_t^{0}]   \nonumber\\
  &\quad + \frac{1}{2}(\sigma^0 \Ex[\pi_t^{m}|\F_t^{0}])^2 + \lambda_t^o\Ex[\log(1-\pi_t^{m})|\F_t^{0}] \bigg\}dt
   + \sigma^0m_t^{X}\Ex[\pi_t^{m}|\F_t^{0}] dW_t^0  \nonumber\\
  &= m_t^{X} \widehat{\eta}(t;\pi^{m}) dt +  \sigma^0 m_t^{X}\Ex[\pi_t^{m}|\F_t^{0}] dW_t^0, \nonumber
\end{align}
where we have used the notation
\begin{align*}
\widehat{\eta}(t;\pi^{m}):=& r+ (b+\lambda_t^{f,o})\Ex[\pi_t^{m}|\F_t^0] - \frac{1}{2}\big(\sigma^2+ (\sigma^0)^2\big)\Ex[(\pi_t^{m})^2|\F_t^{0}]\\
&+ \frac{1}{2}(\sigma^0 \Ex[\pi_t^{m}|\F_t^{0}])^2 + \lambda_t^{f,o}\Ex[\log(1-\pi_t^{m})|\F_t^{0}].
\end{align*}
To solve the FBSDE \eqref{eq:FBSDE-2}, we consider the ansatz that
\begin{equation}\label{eq:P-ansatz-limit}
  P_t^{m} = (X_t^{m})^{\gamma-1}(m_t^{X})^{-\theta\gamma}\varphi_t, \ \  t\in[0,T],
\end{equation}
where $\varphi:[0,T]\to\R$ is a deterministic function of class $C^1$, which satisfies the terminal condition $\varphi_T= 1$. First, note that $P_T^{m} = (X_T^{m})^{\gamma-1}(m_T^{X})^{-\theta\gamma}$ holds trivially. Applying It\^{o}'s lemma to $P_t^{m}$, we can obtain that
\begin{align}\label{eq:adjont-P-Ito}
  &dP_t^{m}=
   P_t^{m}\bigg\{ \frac{\dot{\varphi}_t}{\varphi_t}+ (\gamma-1)\left[r + \left(b+\lambda_t^{f,o}\right)\pi_t^{m}\right] -\theta\gamma\widehat{\eta}(t;\pi^{m})
   -\theta\gamma(\gamma -1)(\sigma^0)^2\pi_t^{m}\Ex[\pi_t^{m}|\F_t^{0}]  \\
 & \quad + \frac{1}{2}(\gamma-1)(\gamma-2)\left( \sigma^2 + (\sigma^0)^2\right)(\pi_t^{m})^2
  + \frac{1}{2}\theta\gamma(\theta\gamma+1)(\sigma^0)^2 \left(\Ex[\pi_t^{m}| \F_t^0]\right)^2\nonumber\\
 & \quad  + \left[ \left(1-\pi_t^{m}\right)^{\gamma-1}-1\right]\lambda_t^{f,o} \bigg\}dt + (\gamma-1)P_t^{m}\pi_t^{m}\sigma dW_t + P_t^{m}\left[(\gamma-1)\pi_t^{m}\sigma^0 - \theta\gamma\sigma^0\Ex[\pi_t^{m}|\F_t^0]\right]dW_t^0  \nonumber\\
 &\quad + P_t^{m}\{(1-\pi_{t}^{m})^{\gamma-1}-1\}dM_t^f.\nonumber
\end{align}
Comparing the expressions of $P_t^{m}$ in \eqref{eq:FBSDE-2} and \eqref{eq:adjont-P-Ito}, we have that
\begin{equation}\label{eq:Q-Y-Q0}
\left\{
\begin{aligned}
 &Q_t^{m} = (\gamma-1)\sigma P_t^{m}\pi_t^{m},\\[0.4em]
 &Y_t^{m} = P_t^{m}\left[\left(1-\pi_{t-}^{m}\right)^{\gamma -1 }-1\right], \\[0.4em]
 &Q_t^{0,m} = P_t^{m} \left\{ (\gamma-1)\pi_t^{m}\sigma^0 - \theta\gamma\sigma^0\Ex[\pi_t^{m}|\F_t^0] \right\}.
\end{aligned}
 \right.
\end{equation}
Let $\pi^{*,m}=(\pi_t^{\ast,m})_{t\in[0,T]}$ be a candidate optimal control that may depend on $m$, and $X^{*,m}=(X_t^{*,m})_{t\in[0,T]}$ be the wealth process under $\pi^{\ast,m}$. For the solution $(P^{*,m}, Q^{*,m}, Q^{0,*,m},Y^{*,m})$ of FBSDE \eqref{eq:FBSDE-2} with $(\pi^m,X^m)$ replaced by $(\pi^{\ast,m},X^{\ast,m})$, we have from \eqref{eq:hamiltion-H-limit} that, for any $\pi^m\in{\cal A}_{\rm MF}$,
\begin{align}\label{eq:Hstart}
  & H(t,X_t^{*,m},\pi^{m},P_t^{*,m}, Q_t^{*,m}, Q_t^{0,*,m}, Y_t^{*,m}) \nonumber\\
  &\qquad= ( rX_t^{*,m} + b \pi^{m}X_t^{*,m})P_t^{*,m}  + \sigma\pi^{m} X_t^{*,m} Q_t^{*,m}
     + \sigma^0\pi^{m} X_t^{*,m}Q_t^{0,*,m} - \pi^{m} X_t^{*,m}\lambda_t^{f,o} P_t^{*,m}, \nonumber\\
  &\qquad = rX_t^{*,m}P_t^{*,m}  + \pi^{m}( b X_t^{*,m} P_t^{*,m}  + \sigma X_t^{*,m} Q_t^{*,m}  + \sigma^0X_t^{*,m}Q_t^{0,*,m} -  X_t^{*,m}\lambda_t^{f,o} P_t^{*,m}).
\end{align}
It can be observed from \eqref{eq:Hstart} that $H$ is linear in $\pi^{m}$. It is then natural to make the coefficient of $\pi^m$ vanish, i.e., for $t\in[0,T]$,
\begin{equation}\label{eq:candidate-conti}
   b P_t^{*,m}  + \sigma Q_t^{*,m}  + \sigma^0Q_t^{0,*,m} - \lambda_t^{f,o} P_t^{*,m} = 0.
\end{equation}
We first apply the relation in \eqref{eq:Q-Y-Q0} to have that
\begin{equation}\label{eq:Q-Y-Q0-star}
\left\{
\begin{aligned}
 &Q_t^{*,m} = (\gamma-1)\sigma P_t^{*,m}\pi_t^{*,m},\\[0.4em]
 &Y_t^{*,m} = P_t^{*,m}\left[(1-\pi_{t}^{*,m})^{\gamma -1}-1\right], \\[0.4em]
 &Q_t^{0,*,m} = P_t^{*,m} \left\{ (\gamma-1)\pi_t^{*,m}\sigma^0 - \theta\gamma\sigma^0\Ex[\pi_t^{*,m}|\F_t^0] \right\}.
\end{aligned}
 \right.
\end{equation}
Plugging \eqref{eq:Q-Y-Q0-star} into \eqref{eq:candidate-conti}, we get that the candidate best response $\pi^{*,m}$ satisfies the equation:
\begin{equation}\label{eq:pi-limit-condidate}
  (\gamma-1)\left[ \sigma^2 + (\sigma^0)^2\right] \pi_t^{\ast,m} -\theta\gamma(\sigma^0)^2\Ex[\pi_t^{\ast,m}|\F_t^0] -\lambda_t^{f,o}\left[(1-\pi_{t}^{\ast,m})^{\gamma-1} -1\right] + b  =  0,~\forall t\in[0,T].
  \end{equation}

Next, we focus on a deterministic MFE and assume that $\pi^{*,m} = (\pi_t^{\ast,m})_{t\in[0,T]}$ is deterministic. Therefore, the condition \eqref{eq:pi-limit-condidate} reduces to
\begin{equation}\label{eq:pi-condidate-determi}
  (\gamma-1)\left[ \sigma^2 + (\sigma^0)^2\right] \pi_t^{\ast,m} -\theta\gamma(\sigma^0)^2\pi_t^{\ast,m} -\lambda_t^{f,o}\left((1-\pi_{t}^{\ast,m})^{\gamma-1} -1\right) + b  =  0.
\end{equation}
As $\lambda_t^{f,o}$ for $t\in[0,T]$ in \eqref{eq:lambda-MF} is deterministic and bounded, by Lemma \ref{lem:Phi}, we can easily get that, for $t\in[0,T]$, there exits a unique $\pi_t^{*,m} \in (0, 1-\epsilon_0 ]$ such that $\Phi(\pi_t^{*,m}, \lambda_t^{f,o}) = 0$. Equivalently, for $t\in[0,T]$, we have that
\begin{equation}\label{eq:pi-MF}
  \pi_t^{*,m} = \phi(\lambda_t^{f,o}) \in (0, 1-\epsilon_0 ],
\end{equation}
where $\phi:\R_+\to\R$ is a Lipschitz continuous function. 
We can easily verify that  $\pi^{*,m}=(\pi_t^{*,m})_{t\in[0,T]} \in {\cal A}_{\rm MF}$. This yields a best (deterministic) response control $\pi_t^{*,m} = \phi(\lambda_t^{f,o})$ for $t\in[0,T]$. We observe that as the coefficients in \eqref{eq:pi-condidate-determi} do not depend on $m\in{\cal D}$ (so is the function $\Phi$ defined by \eqref{eq:Phi}), $\pi^{*,m}$ is independent of $m$. Hence, we can write $\pi^{*,m}$ as $\pi^*$, i.e., $\pi_t^*=\phi(\lambda_t^{f,o})$ for $t\in[0,T]$.

Next, we just need to solve the function $\varphi:[0,T]\to\R$ in the ansatz solution \eqref{eq:P-ansatz-limit} so that the adjoint processes are well defined. Comparing the drift term of $P_t^{*,m}$(associated with $\pi^{*,m}$) in \eqref{eq:FBSDE-2} and \eqref{eq:adjont-P-Ito}, we have that
\begin{align}\label{eq:varphi}
   & P_t^{*,m}\bigg\{ \frac{\dot{\varphi}_t}{\varphi_t}+ (\gamma-1)\left[r + \left(b+\lambda_t^{f,o}\right)\pi_t^{\ast}\right] -\theta\gamma\eta(t;\pi^{*})
   -\theta\gamma(\gamma -1)(\sigma^0)^2(\pi_t^{\ast})^2  \\
 & \qquad\quad + \frac{1}{2}(\gamma-1)(\gamma-2)\left( \sigma^2 + (\sigma^0)^2\right)(\pi_t^{\ast})^2
  + \frac{1}{2}\theta\gamma(\theta\gamma+1)(\sigma^0)^2\left(\pi_t^{\ast}\right)^2\nonumber\\
 & \qquad\quad  + \left[ (1-\pi_t^{*})^{\gamma-1} - 1 \right]\lambda_t^{f,o}\bigg\}   \nonumber\\
    &\quad =  -\left[ rP_t^{*,m} + \pi_t^{\ast}\left(bP_t^{*,m}+\sigma Q_t^{*,m}+ \sigma^0Q_t^{0,*,m}-\lambda_t^{f,o} Y_t^{*,m} \right)\right]. \nonumber
\end{align}
Plugging \eqref{eq:P-ansatz-limit} and \eqref{eq:Q-Y-Q0-star} into \eqref{eq:varphi}, we obtain that
\begin{align}\label{eq:varphi-ODE}
 \frac{\dot{\varphi}_t}{\varphi_t} & = -\gamma r - (\gamma-1)(b+\lambda_t^{f,o})\pi_t^{*}
  + \theta\gamma\widehat\eta(t;\pi^{*}) + \theta\gamma(\gamma -1)(\sigma^0)^2(\pi_t^{*})^2  \nonumber\\
 & \quad - \frac{1}{2}(\gamma-1)(\gamma-2)\left( \sigma^2 + (\sigma^0)^2\right)(\pi_t^{*})^2
  - \frac{1}{2}\theta\gamma(\theta\gamma+1)(\sigma^0)^2 (\pi_t^{*})^2\nonumber\\
 & \quad  - \left[ (1-\pi_t^{*})^{\gamma-1} - 1 \right]\lambda_t^{f,o}\nonumber\\
 &\quad  =: \rho(t)
\end{align}
with the terminal condition $\varphi_T =1$. We stress here that $\rho(t)$ depends on $\pi_t^{*}$ for $t\in[0,T]$. We can then deduce that $\widehat{\eta}(t,\pi^{\ast})=\eta(t; \pi_t^{*} )$, where $\eta(t; \pi)$ for $(t,\pi)\in [0,T]\times U $ is defined by \eqref{eq:eta-new}, i.e.,
\begin{align*}
  \eta(t; \pi ) = r+ (b+\lambda_t^{f,o})\pi - \frac{1}{2}\sigma^2\pi^2 + \lambda_t^{f,o} \log(1-\pi).
\end{align*}
By solving the ODE problem \eqref{eq:varphi-ODE}, we have that
\begin{align}\label{eq:varphisol}
\varphi_t = e^{\int_{t}^{T}\rho(s)ds},~~\forall t\in[0,T].
\end{align}
It then follows from \eqref{eq:FBSDE-2} that the adjoint processes corresponding to $\pi^{*}$  can be rewritten by
\begin{equation}\label{eq:P-Q-Y-solution-MF}
\left\{
  \begin{aligned}
  &P_t^{*,m}  =(X_t^{*})^{\gamma-1}m_t^{-\theta\gamma}e^{\int_{t}^{T}\rho(s)ds},\\
  &Q_t^{*,m} =  (\gamma-1)\sigma\pi_t^{*}(X_t^{*})^{\gamma-1}m_t^{-\theta\gamma}e^{\int_{t}^{T}\rho(s)ds},\\
  &Y_t^{*,m} = \left[(1-\pi_{t}^{*})^{\gamma -1 }-1\right] (X_t^{*,m})^{\gamma-1}m_t^{-\theta\gamma}e^{\int_{t}^{T}\rho(s)ds},\\
  &Q_t^{0,*,m} =  \sigma^0 [(1-\theta)\gamma-1] \pi_t^{*} (X_t^{*})^{\gamma-1}m_t^{-\theta\gamma}e^{\int_{t}^{T}\rho(s)ds},
 \end{aligned}
\right.
\end{equation}
where $X^*= (X_t^*)_{t\in[0,T]}$ is the wealth process under $\pi^* = (\pi_t^*)_{t\in[0,T]}$.

Finally, using the consistency condition in \textit{Step 2} that $m_t^* = \exp\{\Ex[\log(X_t^{\ast})|\F_t^0]\}$ with $t\in [0,T]$, we next derive the expression of $m^* = (m_t^*)_{t\in[0,T]}$. To this purpose, let us recall the process $m^X=(m_t^X)_{t\in[0,T]}$ in \eqref{eq:geometric-m} satisfies $ m_t^{X}:= \exp\{\Ex[\log(X_t^{m})|\F_t^{0}]\}$ for $t\in[0,T]$, where $X^{m}=(X_t^{m})_{t\in[0,T]}$ is the wealth process under an arbitrary strategy $\pi^m\in{\cal A}_{\rm MF}$. Then, we have that $m_t^*=m_t^{X^*}=\exp\{\Ex[\log(X_t^{*})|\F_t^{0}]\}$ for $t\in[0,T]$ and it follows that $m^*=(m_t^*)_{t\in[0,T]}$ is given by \eqref{eq:geometric-m-thm}. We therefore conclude that $\pi_t^* = \phi(\lambda_t^{f,o})$ for $t\in[0,T]$ is a  deterministic MFE, which completes the proof.
\end{proof}

\begin{remark}
We emphasize that the assumption $\bm{(A_O)}$ with a constant limiting type vector $o\in {\cO}$ is needed to guarantee the existence of a deterministic MFE strategy in \autoref{thm:optimal-MFG} in the model with both common noise $W^0$ and contagious jump risk. We focus on a deterministic MFE strategy not only because it exhibits clean and interpretable analytical form, but it also crucially simplifies some future proofs to show the validity of a contructed approximate Nash equilibrium in the $n$-player game and to analyze its explicit convergence rate.
\end{remark}

\section{Discussions of the Mean Field Equilibrium}\label{sec:numerical}

We present in this section some quantitative properties and sensitivity results of the deterministic MFE strategy $\pi_t^* = \phi(\lambda_t^{f,o})$ obtained in \autoref{thm:optimal-MFG}. First, Lemma~\ref{quanprop-1} (proved in \ref{sec:proofs}) summarizes some monotonicity results on several model parameters.
\begin{lemma}\label{quanprop-1}
For each fixed $t\in[0,T]$, let us use the notation $\pi_t^*(b,\sigma,\sigma^0,\gamma,\theta)$ to highlight the dependence of the deterministic MFE strategy $\pi^*$ on the model parameters $(b,\sigma,\sigma^0,\gamma,\theta)$. Then, we have that
\begin{itemize}
\item[{\rm(i)}] $b\mapsto\pi_t^*(b,\sigma,\sigma^0,\gamma,\theta)$ is increasing;
\item[{\rm(ii)}] both $\sigma\mapsto\pi_t^*(b,\sigma,\sigma^0,\gamma,\theta)$ and $\sigma^0\mapsto\pi_t^*(b,\sigma,\sigma^0,\gamma,\theta)$ are decreasing;
\item[{\rm(iii)}] both $\gamma\mapsto\pi_t^*(b,\sigma,\sigma^0,\gamma,\theta)$ and $\theta\mapsto\pi_t^*(b,\sigma,\sigma^0,\gamma,\theta)$ are decreasing.
\end{itemize}
\end{lemma}

Note that items (i) and (ii) in Lemma \ref{quanprop-1} are consistent with our intuition that the higher return and lower volatility in the limiting market model will incentivize the agent to invest more in the risky asset account in the mean field equilibrium strategy. When the representative agent is more risk averse, item (iii) implies that the representative agent becomes more conservative and invests less in the risky asset. It is also interesting to see that for $0<\gamma<1$, the mean field Nash equilibrium $\pi^*$ has no short-selling and a higher competition parameter $\theta$ leads to a lower investment proportion in the risky asset. That is, in the equilibrium state, a more competitive agent with the risk aversion $0<\gamma<1$ will prefer to invest more in the riskless bond account. Here, we note that both $\gamma$ and $\theta$ are parameters for the whole population (as all agents are symmetric in the mean field model). If the risk tolerance of the population is high with ($\gamma\in(0,1)$ for all agents, when the whole population becomes more competitive as $\theta$ increases, the representative agent actually behaves less competitively in a highly competitive environment. This can be explained via a game theoretical thinking. Suppose other agents initially hold large long positions in the stock, there are essentially two strategies for the representative agent: $(i)$ She can also hold a large long position in the stock and hope to outperform other peers when the price goes up. However, as all peers are very competitive and allocate large amount of wealth into the stock, the chance to outperform others is actually slim because other peers may receive even higher wealth return. $(ii)$ Or the representative agent can reduce her allocation in the stock and take advantage of the risk that the stock price may diffuse down or jump downward frequently such that her terminal wealth in riskless asset can significantly outperform other peers whose wealth drop due to default risk. If the representative agent can tolerate some risk and is highly concerned with the relative performance, she will adopt the second strategy, leading to the reduction of investment in the stock. Similarly, as all other peers are as competitive as the representative agent, all of them will also reduce the investment in the stock and aim to outperform others when the default jump occurs or the price goes down. Therefore, for $\lambda\in (0,1)$, a higher $\theta$ actually leads to a lower equilibrium allocation in the stock in this mean field game.

We next numerically illustrate the sensitivity results of $\pi^*$ with respect to jump risk parameters. We first note that the jump contagion effect among all stocks becomes negligible in the mean field model. However, the jump intensity process in the mean field model comes from the model with contagion effect in the $n$-player game model. That is, the larger contagion effect among stocks in the $n$-player model with larger $\beta_i$ and $\varsigma_i$ will lead to a larger jump intensity process $\lambda_t^{f,o}$ in the mean field model. Therefore, our numerical examples can partially reflect how the contagion effect in the $n$-player game affects the equilibrium behavior when there are infinitely many agents. Recall that the function $\phi(\lambda)$ is decreasing in $\lambda>0$ by Lemma \ref{lem:Phi}. Here, we take a differentiable and Lipschitz function $f$ such that $f(x)= x $ when $x \leq M $, while $f(x)= M + \delta_0 $ when $x > M+\delta_0$ for positive constants $M,\delta_0$. It is also clear that the limiting intensity factor $\lambda_t^l$ admitting the explicit form
$ \lambda_t^{l} = \frac{\alpha\lambda_{\infty}}{\alpha-\beta\varsigma} +(\lambda_0 -\frac{\alpha\lambda_{\infty}}{\alpha-\beta\varsigma})e^{(\beta\varsigma-\alpha)t}$ when $\sup_{t\in[0,T]}\lambda_t^l\leq M$ (this can be achieved by taking appropriate parameter values listed in Table~\ref{table:parameters} together with $M$ large enough). In the following numerical simulation, we take $M = 10^2$ and $\delta_0 = 10^{-2} $. It is increasing in time $t$ if $\beta\varsigma>\alpha(1-\frac{\lambda_{\infty}}{\lambda_0})$ and it is decreasing in time $t$ otherwise. Thanks to the analytical structure of $\pi_t^*=\phi(\lambda_t^{f,o})$, if model parameters satisfy that $\beta\varsigma>\alpha(1-\frac{\lambda_{\infty}}{\lambda_0})$, the deterministic MFE strategy $\pi_t^*$ is decreasing in time $t$ indicating that the representative agent in the mean field game will reduce the portfolio in the risky asset as time evolves because of the increasing probability of the downward jump risk. To numerically illustrate this case, we choose parameters from Table~\ref{table:parameters} except the time variable $t$. We then plot the function  $t\rightarrow \pi^*_t$ in Figure \ref{pi-star-t}, which is shown to be a decreasing function of $t\in[0, 10]$.

\begin{table}[h!]\footnotesize\tabcolsep 16pt
\begin{center}
\caption{The chosen parameter values}\label{table:parameters}\vspace{-2mm}
\end{center}
\begin{center}
\begin{tabular}{ccc}
  Model Parameter & Financial Meaning & Value \\\hline
  $\epsilon_0\in(0,1)$ & {$1-\epsilon_0$ is the upper bound of strategy} & $10^{-10}$\\
$\gamma\in(0,1)$ & {$1-\gamma$ is the degree of relative risk aversion} & $0.4$\\
$\theta\in[0,1]$ & {competition weight} & 0.5\\
$\sigma>0$ & {limiting volatility of stock price} & 0.3\\
$\sigma_0>0$ & {common volatility} & 0.2 \\
$b\in\R$ & return premium of stock & 0.2\\
$\lambda_0>0$ & initial default intensity & 0.1\\
$\lambda_{\infty}>0$ & long-term default intensity level & 0.6\\
$\alpha>0$ & adjustment speed of intensity toward long-term level & 0.5\\
$\beta>0$ & limiting jump weight of default intensity & 0.4\\
$\varsigma>0$ & limiting jump weight of default intensity & 0.2\\
$t>0$ & the current time level & 3\\
\end{tabular}
\end{center}
\end{table}

With the parameter values given in Table~\ref{table:parameters}, another interesting observation is that the limiting intensity process converges to a long-run steady level as $t$ tends to $+\infty$ when the parameters satisfy $\beta\varsigma<\alpha$. In this case, it follows from the explicit solution of $\lambda_t^{f,o}$  that
\begin{align}\label{eq:limitlambdat}
 \lim_{t\to\infty}\lambda_t^{f,o} = \lim_{t\to\infty}f(\lambda_t^l) =\frac{\alpha\lambda_{\infty}}{\alpha-\beta\varsigma}.
\end{align}
Thus, using the continuity of $\lambda\mapsto\phi(\lambda)$ (see Lemma~\ref{lem:Phi}), we can also derive the long-run behavior of the deterministic MFE (given the time horizon $T$ is sufficiently large) that
\begin{align}\label{eq:limitphilambdat}
 \lim_{t\to\infty}\pi_t^*=\lim_{t\to\infty}\phi(\lambda_t^{f,o})=\phi\left(\frac{\alpha\lambda_{\infty}}{\alpha-\beta\varsigma}\right).
\end{align}
For the given parameters, it is observed from Figure  \ref{pi-star-t} that this long run steady value of the deterministic MFE is approximately $0.31$.

\begin{figure}[h!]
\centering
\includegraphics[width=0.5\textwidth]{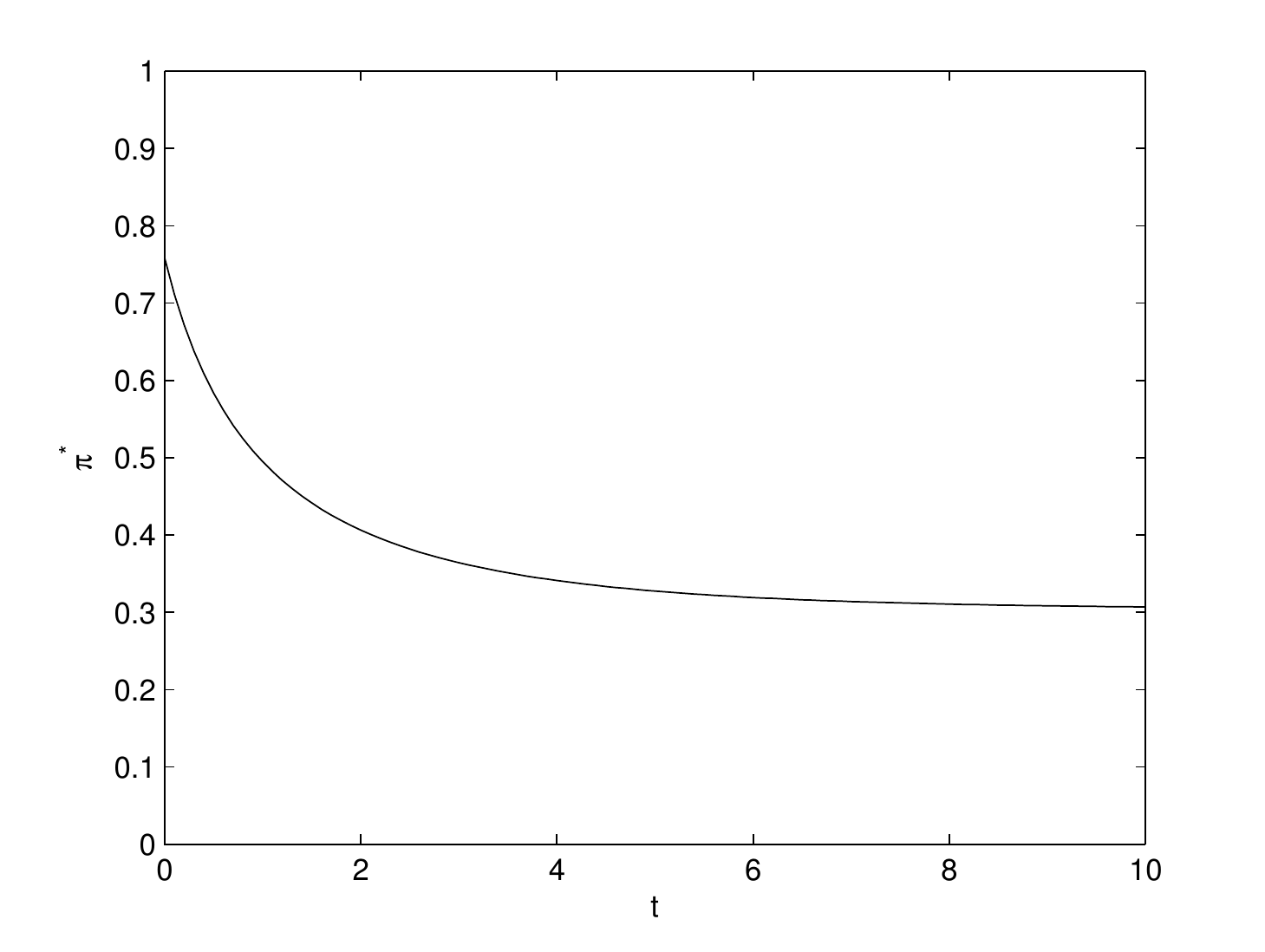}
\caption{The deterministic MFE $\pi^*$ as a function of time $t$.}
\label{pi-star-t}
\end{figure}

Let us then examine the sensitivity of the mean field equilibrium $\pi^*$ w.r.t. the mean recovery speed parameter $\alpha$ in the mean field model. To this end, we choose and fix parameters as in Table~\ref{table:parameters} apart from the parameter $\alpha$, and plot the function of $ \alpha\rightarrow \pi^*_t $ in terms of the parameter $\alpha$ in Figure \ref{pi-star-alpha} on the interval $[0.1, 1]$. For the given $t=3$, we can observe that the equilibrium $\pi^*$ is decreasing in the parameter $\alpha$.   For the chosen parameters, it is easy to see that a larger $\alpha$ leads to a larger default intensity of the risky asset. Consequently, to avoid the higher probability of default, the agent prefers to invest less in the risky asset account. Similarly, we plot in Figure \ref{pi-star-beta} the mean filed equilibrium $\beta\rightarrow \pi_t^*$ as a function in terms of the parameter $\beta$ from the intensity process $\lambda_t^{f,o}$. We choose and fix other parameters as in Table~\ref{table:parameters} apart from the parameter $\beta$. For the given $t=3$, the mean field equilibrium $ \pi_t^*(\beta)$ is decreasing in the parameter $\beta$. We also note that $\beta$ and $\varsigma$ are symmetric in the definition of $\lambda_t^{f,o}$, the sensitivity result of the mean filed equilibrium $\pi_t^*(\varsigma)$ w.r.t. the parameter $\varsigma$ is similar to the case w.r.t. the parameter $\beta$. Again, from the definition of $\lambda_t^{f,o}$, one can see that the intensity value is increasing in terms of $\beta$ and $\varsigma$. Therefore, as $\beta$ or $\varsigma$ increases, the agent invests less in the risky asset due to the higher probability of default.

\begin{figure}[h!]
\centering
\includegraphics[width=0.5\textwidth]{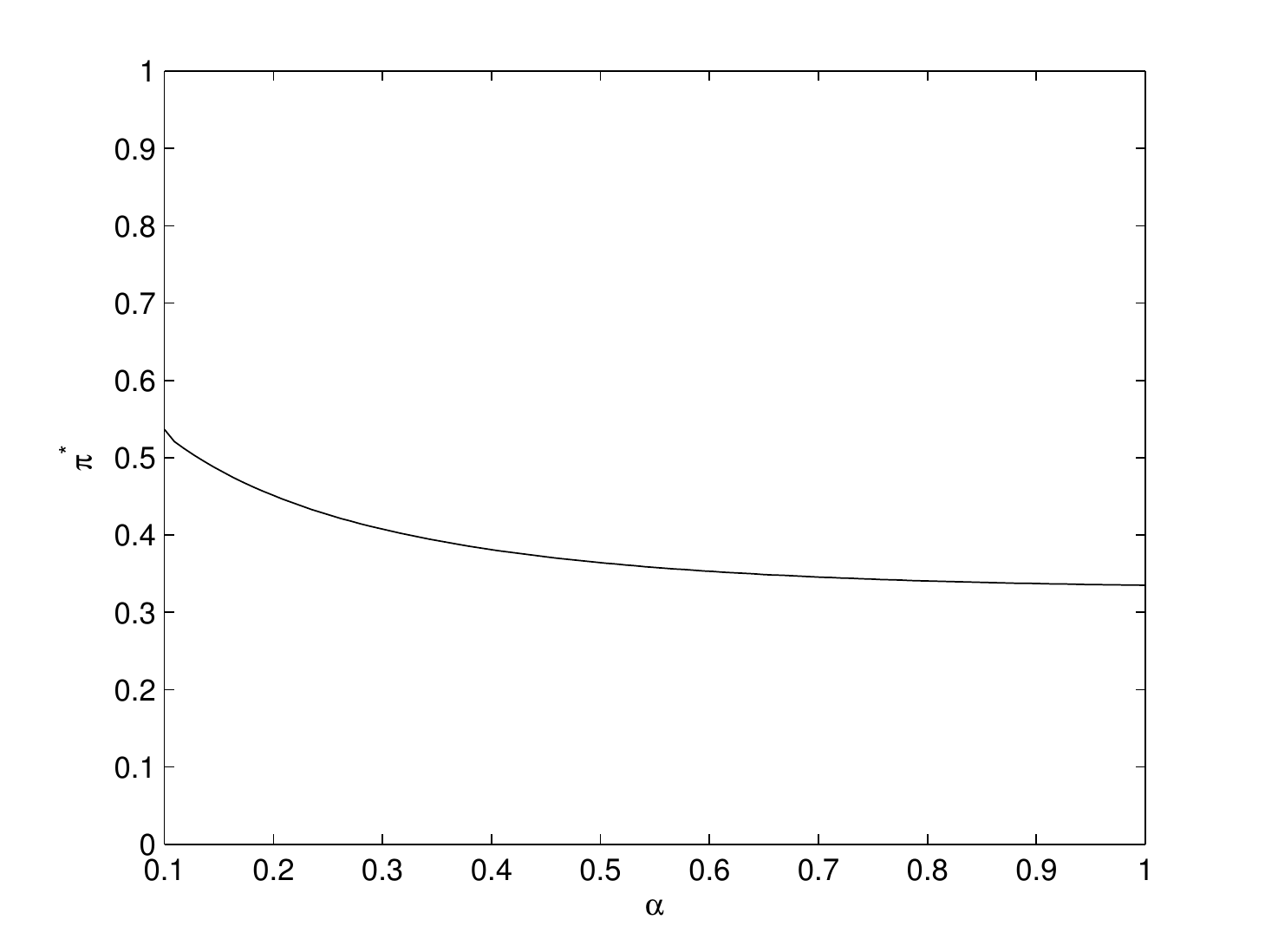}
\caption{{\small The sensitivity of the deterministic MFE $\pi^*$ w.r.t. the parameter $\alpha$.}}
\label{pi-star-alpha}
\end{figure}

\begin{figure}[h!]
\centering
\includegraphics[width=0.5\textwidth]{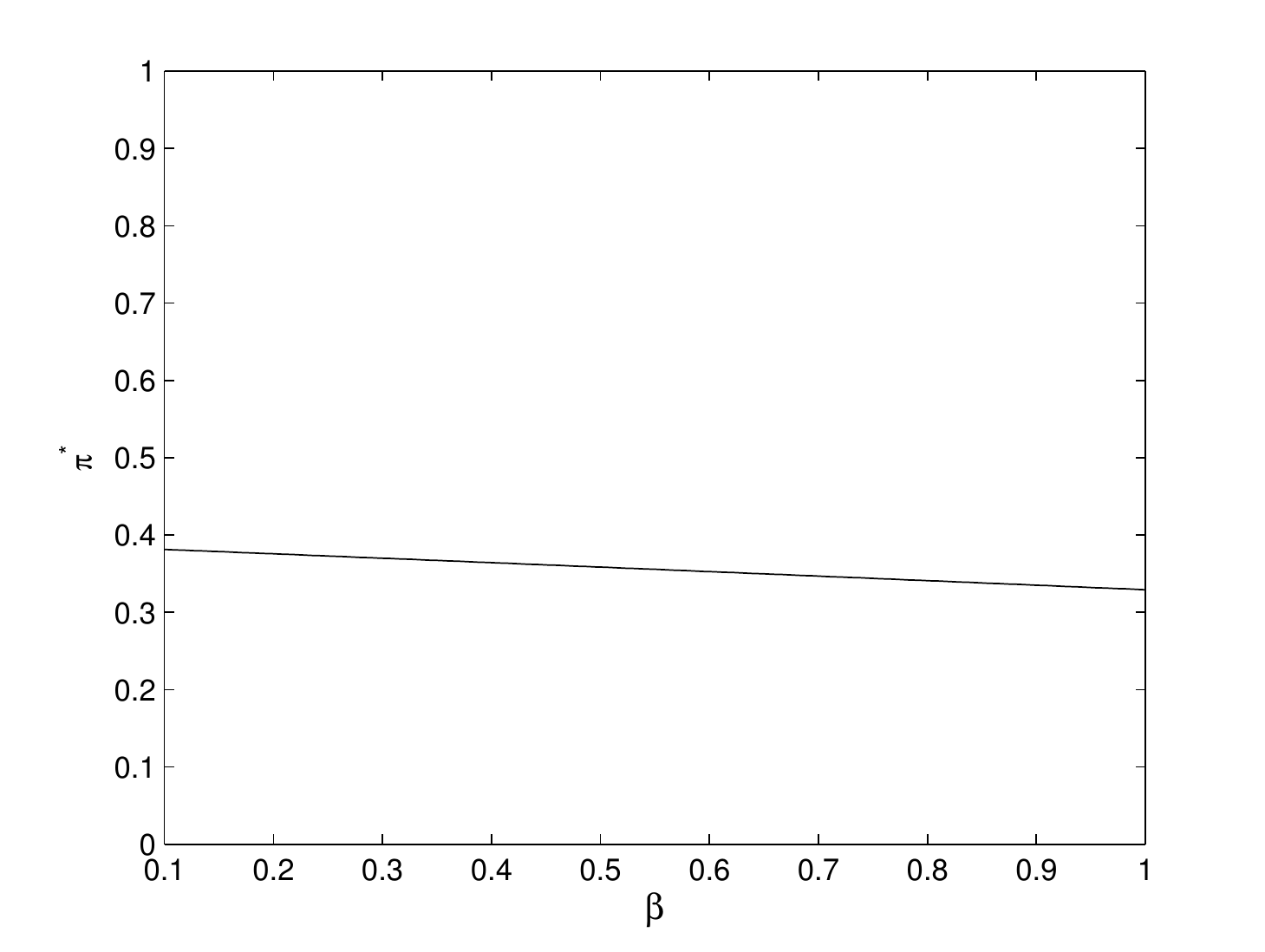}
\caption{{\small The sensitivity of the deterministic MFE $\pi^*$ w.r.t. the parameter $\beta$.}}
\label{pi-star-beta}
\end{figure}

\section{Approximate Nash Equilibrium in the $n$-Player Game}\label{sec:epsilon-Nash-equilibrium}

The goal of this section is to show that the mean field equilibrium obtained in \autoref{thm:optimal-MFG} can help us to construct an approximate Nash equilibrium in the game with a sufficiently large but finite number $n$ of agents. Furthermore, the explicit order of the approximation error can also be derived, which will facilitate the practical implementations  of the mean field approximation in finite population game applications.

Recall that the intensity process $ (\Lambda_t^{f,i,n})_{t\in[0,T]}$ follows the dynamics that
\begin{equation}\label{eq:lambda-i-0}
  \Lambda_t^{f,i,n} = f(\lambda_t^{i,n}),\quad d\lambda_t^{i,n} = \alpha_i(\lambda_{\infty}^i - \lambda_t^{i,n} )dt + \frac{\beta_i}{n}\sum_{j=1}^{n}\varsigma_{j}dN^{j,n}_t, 
  \quad i=1,\ldots,n,
\end{equation}
with the speed parameter $\alpha_i >0$, the mean-reverting level $\lambda_{\infty}^i >0 $, and the jump risk contagion effect parameter $\beta_i\varsigma_j >0$. To avoid the possible ambiguity in notation, we will keep the superscript $n$ in this section. Hereafter, $\Xi_{\Lambda}\subset\R_+$ denotes the state space of $\Lambda^{f,i,n}=(\Lambda_t^{f,i,n})_{t\in[0,T]}$, which is a bounded set and is independent of $(i,n)$ since the jump rate function $f$ is bounded (c.f. \eqref{eq:lambda-i-0}).

Next, before constructing an $\varepsilon$-Nash equilibrium in the $n$-player game, we first introduce an auxiliary problem based on the limiting case, whose best response strategy will help us to construct an approximate Nash equilibrium. Let us define the auxiliary control problem $(\bf{P_n}$) by
\begin{equation}\label{eq:objective-J-bar-n}
 \sup_{\pi^i \in \mathcal{A}^i} \bar{J}_i^n(\pi^i;m^*) :=  \sup_{\pi^i \in \mathcal{A}^i} \Ex\left[\frac{1}{\gamma_i}(X_T^{i,n})^{\gamma_i}(m_T^*)^{-\theta_i\gamma_i} \right], \quad i= 1,\ldots,n,
\end{equation}
subjecting to
\begin{equation}\label{eq:wealth-X-i-P-n}
  \left\{
  \begin{aligned}
  & d X_t^{i,n} = X_t^{i,n}( r + b_i\pi_t^i)dt + \pi_t^i X_{t-}^{i,n}(\sigma_idW_t^i  + \sigma_i^0dW_t^0 - dM_t^{f,i,n}), \\
  & \Lambda_t^{f,i,n} = f(\lambda_t^{i,n}),\quad d\lambda_t^{i,n} = \alpha_i(\lambda_{\infty}^i - \lambda_t^{i,n})dt + \frac{\beta_i}{n}\sum_{j=1}^{n}\varsigma_{j}dN_t^{j,n},\\
  & dm_t^* = m_t^*\eta(t;\pi_t^{*}) dt + \sigma^0 m_t^{*}\pi_t^{*}dW_t^0.
 \end{aligned}
 \right.
\end{equation}
Here, we recall that $m^*=(m_t^*)_{t\in[0,T]}$ is the fixed point given by \eqref{eq:geometric-m-thm}, $\pi^{\ast}=(\pi^*_t)_{t\in[0,T]}\in{\cal A}_{\rm MF}$ is the deterministic MFE characterized by \eqref{eq:pi-mean-limit-thm} in \autoref{thm:optimal-MFG}, and $M_t^{f,i,n} = N_t^{i,n}-\int_0^t\Lambda_s^{f,i,n} ds$ is a $(\Px,\Gx)$-martingale for $i=1,\ldots n$. The following lemma characterizes the optimal strategy of the auxiliary control problem $(\bf{P_n})$. The proof is reported in \ref{sec:proofs}.

\begin{lemma}\label{lem:optimum-Pn}
Let $\tilde{\pi}^{i,n}=\tilde{\pi}^{i,n}(t,{\bm\lambda})\in(-\infty,1)$ for $(t,{\bm\lambda})=(t,(\lambda_1,\ldots,\lambda_n)^{\top})\in[0,T]\times\Xi_{\Lambda}^n$ be the optimal (feedback) strategy of the auxiliary control problem $(\bf{P_n})$. Then, we have that
\begin{align}\label{eq:FOC-22-PP0}
\Phi_i(t,\lambda_i,\tilde{\pi}^{i,n})=  O\left(\frac{1}{n}\right),
\end{align}
where the function $\Phi_i(t,\lambda,\pi): [0,T]\times \Xi_{\Lambda}  \times  (-\infty, 1) \rightarrow\R$ is defined by
\begin{align}\label{eq:Phi-i}
  \Phi_i(t,\lambda,\pi): = (\gamma_i-1)[ \sigma_i^2 + (\sigma_i^0)^2] \pi -  \theta_i\gamma_i\sigma_i^0\sigma^0 \pi_t^* - \lambda(1-\pi )^{\gamma_i -1} + \lambda + b_i,
\end{align}
and $\pi^*=(\pi_t^*)_{t\in[0,T]}$ is the deterministic MFE given in \autoref{thm:optimal-MFG}. Moreover, there exists a pair $(\tilde{D},\tilde{\epsilon})\in\R\times(0,1)$ independent of $(t,{\bm\lambda},i,n)$ such that $\tilde{\pi}^{i,n}=\tilde{\pi}^{i,n}(t,{\bm\lambda})\in[\tilde{D},1-\tilde{\epsilon}]$.
\end{lemma}

Intuitively, it can be seen from \eqref{eq:FOC-22-PP0} in Lemma~\ref{lem:optimum-Pn} that, as $n$ tends to infinity, the optimal feedback strategy $\tilde{\pi}^{i,n}=\tilde{\pi}^{i,n}(t,{\bm\lambda})$ converges to some $\hat{\pi}^i$, where $\hat{\pi}^i$ satisfies $\Phi_i(t,\lambda_i,\hat{\pi}^i)=0$. The following lemma characterizes the zero point of $\pi\mapsto\Phi_i(t,\lambda,\pi)$. The proof is similar to that of Lemma~\ref{lem:Phi}, and we hence omit it.
\begin{lemma}\label{lem:map-Phi-i}
For the function $\Phi_i(t,\lambda,\pi)$ defined by \eqref{eq:Phi-i}, we have that, for any $(t,\lambda)\in[0,T]\times\Xi_{\Lambda}$, there exists a unique $ \hat{\pi} =\hat{\pi}(t,\lambda)\in [D_0,1-\epsilon_0]$ such that $\Phi_i(t,\lambda,\hat{\pi})=0$ where $(D_0,\epsilon_0)\in\R\times(0,1)$ is a pair of constants that are independent of $(t,\lambda,i,n)$. Moreover, there exists a unique continuous function $\phi_i:[0,T]\times\Xi_{\Lambda}\to\R$ such that
\begin{equation}\label{eq:a-ast-lem}
  \hat{\pi} = \phi_i(t,\lambda), \ \ (t,\lambda) \in [0,T]\times\Xi_{\Lambda},
\end{equation}
where $\phi_i$ also has a continuous partial derivative with respect to $\lambda$.
\end{lemma}

The above function $\phi_i$ in \eqref{eq:a-ast-lem} also plays an important role in the construction of an approximating Nash equilibrium. More precisely, for $ i=1,\ldots,n$, we consider a strategy of the agent $i$ that
\begin{equation}\label{eq:pi-epsilon-i}
  \pi_t^{*,i,n} := \phi_i(t,\Lambda_t^{f,i,n}),\quad t\in[0,T],
\end{equation}
where $\Lambda^{f,i,n}=(\Lambda_t^{f,i,n})_{t\in[0,T]}$ is given in \eqref{eq:wealth-X-i-P-n}. It follows from Lemma \ref{lem:map-Phi-i} that $\pi_t^{*,i,n} = \phi_i(t,\Lambda_t^{f,i,n}) \in [D_0, 1-\epsilon_0]$ for some $(D_0,\epsilon_0)\in\R\times(0,1)$ in view that $\Lambda^{f,i,n}\in\Xi_{\Lambda}$ is bounded. Then, the corresponding wealth process $ X^{*,i} =  (X_t^{*,i})_{t\in[0,T]} $ of agent $i$ under the strategy $\pi^{\ast,i,n}=(\pi^{*,i,n}_t)_{t\in[0,T]}$ is governed by
 \begin{align}\label{eq:wealth-X-i-ast}
  d X_t^{\ast,i,n}
  & = X_t^{\ast,i,n}( r + b_i\pi_t^{\ast,i,n})dt + \pi_t^{\ast,i,n}X_{t-}^{\ast,i,n}(\sigma_idW_t^i  + \sigma_i^0 dW_t^0  -dM_t^{f,i,n}).
\end{align}

\begin{remark}
Comparing $\Phi_i(t,\lambda,\pi)$ in \eqref{eq:Phi-i} for the $n$-player game and $ \Phi(\pi,\lambda)$ in \eqref{lem:Phi} for the MFG, we note that the approximate Nash equilibrium $\pi_t^{*,i,n}$ in \eqref{eq:pi-epsilon-i} is constructed specifically based on the analytical form of the mean field equilibrium $\pi^*$ when parameters $(b,\sigma,\sigma^0,\gamma)$ in $ \Phi(\pi,\lambda)$ are modified to $(b_i,\sigma_i,\sigma_i^0,\gamma_i)$ and the term $\theta\gamma(\sigma^0)^2\pi $ is replaced by the term $ \theta_i\gamma_i\sigma_i^0\sigma^0 \pi_t^*$ depending on the mean field equilibrium $\pi^*$. That is, we propose a construction of $\pi_t^{*,i,n}$ based on the mean field equilibrium $\pi^*$ both implicitly (via the analytical form of $ \Phi(\pi,\lambda)$ in \eqref{lem:Phi}) and explicitly (via the term $ \theta_i\gamma_i\sigma_i^0\sigma^0 \pi_t^*$ in \eqref{eq:Phi-i}).
\end{remark}

We can now present the main result of this section, which gives the approximate Nash equilibrium for the $n$-player game.
\begin{theorem}\label{thm:epsilon-Nash-equilibrium}
Let the assumption $\bm{(A_{O})}$ hold. Consider the objective function \eqref{eq:objective} for the $i$th agent.
Then $ \bm{\pi}^{\ast,n} = (\pi_t^{\ast,1,n},\ldots, \pi_t^{\ast,n,n})_{t\in[0,T]}$ is an $ \varepsilon_n$-Nash equilibrium, where $\pi^{\ast,i,n}=(\pi_t^{\ast,i,n})_{t\in[0,T]}$ is given by \eqref{eq:pi-epsilon-i} for $i=1,\ldots, n$. That is, we have that
\begin{equation}\label{eq:thm-epsilon-Nash}
   \sup_{\pi^i \in \mathcal{A}^i} J_i(\pi^{i},\pi^{\ast,-i,n}) \leq J_i(\pi^{\ast,i,n}, \pi^{\ast,-i,n}) + \varepsilon_n,
\end{equation}
where $\pi^{\ast,-i,n}:=(\pi^{*,1,n},\ldots,\pi^{*,i-1,n},\pi^{*,i+1,n},\ldots, \pi^{*,n,n})$. Moreover, the order of the error term satisfies $ \varepsilon_n = O(n^{-\frac{1}{4}})$.
\end{theorem}

To prove \autoref{thm:epsilon-Nash-equilibrium}, we first introduce the following auxiliary results, whose proofs are given in \ref{sec:proofs}.
\begin{lemma}\label{lem:auxiliarylemmas}
Let the assumption $\bm{(A_O)}$ hold. Then, we have that
\begin{itemize}
\item[{\rm(i)}] Consider $\pi^{*,i,n}=(\pi_t^{\ast,i,n})_{t\in[0,T]} $ defined in \eqref{eq:pi-epsilon-i} and $X^{*,i,n}=(X_t^{\ast,i,n})_{t\in[0,T]}$ given by \eqref{eq:wealth-X-i-ast} for $i = 1,\ldots,n$. For any $p\in\R$, there exists a constant $D_p$ independent of $(i,n)$ such that
    \begin{equation}\label{eq:bound-X-i-ast}
  \sup_{t\in[0,T]}\Ex\left[\left|X_t^{\ast,i,n}\right|^p\right] \leq D_p.
\end{equation}
\item[{\rm(ii)}] Let $\bar{\Lambda}^{f,n}_t := \frac{1}{n}\sum_{i=1}^{n} \Lambda_t^{f,i,n}$ for $t\in[0,T]$, and the mean field intensity $\lambda^{f,o}=(\lambda_t^{f,o})_{t\in[0,T]}$ be given by \eqref{eq:lambda-MF}. We have that
    \begin{equation}\label{eq:lem-lambda}
   \sup_{t\in[0,T]}\Ex\left[\left|\bar{\Lambda}_t^{f,n} -\lambda_t^{f,o}\right|^2\right] = O\left(\frac{1}{n^2}\right).
\end{equation}
\item[{\rm(iii)}] Denote by $\bar{X}_t^{*,n}=(\prod_{i=1}^n X_t^{\ast,i,n})^{\frac{1}{n}}$ the geometric mean wealth process at $t\in[0,T]$, where $X^{*,i,n}=(X_t^{\ast,i,n})_{t\in[0,T]}$ is defined in  \eqref{eq:wealth-X-i-ast}, and $m^*=(m_t^{*})_{t\in[0,T]}$ in the limiting model is obtained in \autoref{thm:optimal-MFG}. We have that, for all $t\in [0,T]$,
\begin{equation}\label{eq:lem-bar-X-m}
  \sup_{t\in[0,T]}\Ex\left[\left|\bar{X}_t^{*,n} - m_t^{*}\right|^2\right] = O\left(\frac{1}{\sqrt{n}}\right).
\end{equation}
\end{itemize}
\end{lemma}

Based on all previous preparations, we can give the proof of \autoref{thm:epsilon-Nash-equilibrium}.
\begin{proof}[Proof of \autoref{thm:epsilon-Nash-equilibrium}]
Recall that $X^{*,i,n}=(X_t^{*,i,n})_{t\in[0,T]}$ is defined by \eqref{eq:wealth-X-i-ast}. For ease of presentation, let us denote
\begin{equation*}
  \bar{X}^{\ast,-i,n}_t := \left(X_t^{i,n} \prod_{j\neq i} X^{\ast,j,n}\right)^{\frac{1}{n}},\quad
 \bar{X}^{\ast, n}_t := \left(\prod_{j=1}^n X^{\ast,j,n}\right)^{\frac{1}{n}},\quad t\in[0,T].
\end{equation*}
Here, $X^{i,n}=(X_t^{i,n})_{t\in[0,T]}$ satisfies the wealth dynamics \eqref{eq:wealth-dynamics-i} under the strategy $\pi^{i}=(\pi_t^{i})_{t\in[0,T]}\in{\cal A}^i$. We recall the objective functional $J_i$ defined by \eqref{eq:objective}. Then, we have that
\begin{equation*}
    J_i(\pi^{i},\pi^{\ast,-i,n}) = \Ex\left[\frac{1}{\gamma_i}\left(X_T^{i,n}\right)^{\gamma_i} \left(\bar{X}_T^{\ast,-i,n}\right)^{-\theta_i\gamma_i} \right],\quad J_i({\bm\pi}^{*,n}) = \Ex\left[\frac{1}{\gamma_i}\left(X_T^{\ast,i,n}\right)^{\gamma_i} \left(\bar{X}_T^{\ast,n}\right)^{-\theta_i\gamma_i} \right].
\end{equation*}

Next, we proceed to prove \eqref{eq:thm-epsilon-Nash} by using the introduced auxiliary problem $(\bf{P_n})$ in \eqref{eq:objective-J-bar-n}-\eqref{eq:wealth-X-i-P-n}. Note that
\begin{align}\label{eq:J-J}
 & \sup_{\pi^i \in \mathcal{A}^i} J_i\left(\pi^{i},\pi^{\ast,-i,n}\right) - J_i(\pi^{*,i,n}, \pi^{\ast,-i,n})\notag\\
 &\qquad =  \left(\sup_{\pi^i \in \mathcal{A}^i} J_i(\pi^{i},\pi^{\ast,-i,n}) - \sup_{\pi^i \in \mathcal{A}^i}\bar{J}_i^n(\pi^i;m^*) \right)  + \left(\sup_{\pi^i \in \mathcal{A}^i} \bar{J}_i^n(\pi^i;m^*)  - J_i(\pi^{*,i,n}, \pi^{*,-i,n})\right)\notag\\
  &\qquad\leq \sup_{\pi^i \in \mathcal{A}^i} \left(J_i(\pi^{i},\pi^{\ast,-i,n}) -\bar{J}_i^n(\pi^i;m^*)\right)  + \left(\sup_{\pi^i \in \mathcal{A}^i} \bar{J}_i^n(\pi^i;m^*)  - J_i(\pi^{*,i,n}, \pi^{*,-i,n})\right).
\end{align}
For the first term of RHS of \eqref{eq:J-J}, we have that
\begin{align*}
& J_i(\pi^{i},\pi^{\ast,-i,n}) -\bar{J}_i^n(\pi^{i};m^*)= \Ex\left[\frac{1}{\gamma_i}\left(X_T^{i,n}\right)^{\gamma_i} \left(\bar{X}_T^{\ast,-i,n}\right)^{-\theta_i\gamma_i} \right]
 -  \Ex\left[\frac{1}{\gamma_i}\left(X_T^{i,n}\right)^{\gamma_i} (m_T^*)^{-\theta_i\gamma_i} \right]\\
&\qquad \leq \left\{\Ex\left[\frac{1}{\gamma_i}(X_T^{i,n})^{\gamma_i} (\bar{X}_T^{\ast,-i,n})^{-\theta_i\gamma_i} \right]
 -  \Ex\left[\frac{1}{\gamma_i}(X_T^{i,n})^{\gamma_i}(\bar{X}_T^{\ast,n})^{-\theta_i\gamma_i} \right]\right\}\\
 & \qquad\quad + \left\{\Ex\left[\frac{1}{\gamma_i}\left(X_T^{i,n}\right)^{\gamma_i} \left(\bar{X}_T^{\ast,n}\right)^{-\theta_i\gamma_i} \right]
 -  \Ex\left[\frac{1}{\gamma_i} \left(X_T^{i,n}\right)^{\gamma_i} (m_T^*)^{-\theta_i\gamma_i} \right]\right\}:= I_1^i + I_2^i.
\end{align*}
For the term $I_1^i$, we deduce that
\begin{align*}
 &\Ex\left[\frac{1}{\gamma_i}\left(X_T^{i,n}\right)^{\gamma_i} \left(\bar{X}_T^{\ast,-i,n}\right)^{-\theta_i\gamma_i} \right]
 -\Ex\left[\frac{1}{\gamma_i}\left(X_T^{i,n}\right)^{\gamma_i} \left(\bar{X}_T^{\ast,n}\right)^{-\theta_i\gamma_i} \right]\nonumber\\
 &\qquad
 = \Ex\left[\frac{1}{\gamma_i}\left(X_T^{i,n}\right)^{\gamma_i} \left((\bar{X}_T^{\ast,-i,n})^{-\theta_i\gamma_i}
 -\left(\bar{X}_T^{\ast,n}\right)^{-\theta_i\gamma_i} \right)\right].
\end{align*}
Note that $\gamma_i\in (0,1)$ implies that $p_i := \theta_i\gamma_i \in[0,1)$. Using the inequality $|a^{p_i}- b^{p_i}| \leq p_i |a-b|\max\{a^{p_i-1} ,b^{p_i-1}\}$ for all $a,b >0$, we can derive that
\begin{align*}
 &\left|\left(\bar{X}_T^{\ast,-i,n}\right)^{-\theta_i\gamma_i} -\left(\bar{X}_T^{\ast,n}\right)^{-\theta_i\gamma_i} \right|
 = \frac{1}{(\bar{X}_T^{\ast,-i,n})^{p_i} \left(\bar{X}_T^{\ast,n}\right)^{p_i}} \left| \left(\bar{X}_T^{\ast,n}\right)^{p_i} - \left(\bar{X}_T^{\ast,-i,n}\right)^{p_i}\right| \nonumber\\
 &\qquad= \frac{\left(\prod_{j\neq i}X_T^{\ast,j}\right)^{\frac{p_i}{n}}}{\left(\bar{X}_T^{\ast,-i,n}\right)^{p_i} \left(\bar{X}_T^{\ast,n}\right)^{p_i}} \left|\left(X_T^{\ast,i,n}\right)^{\frac{p_i}{n}} - \left(X_T^{i,n}\right)^{\frac{p_i}{n}}\right| \\
 &\qquad \leq \frac{p_i}{n}\frac{\left(\prod_{j\neq i}X_T^{\ast,j,n}\right)^{\frac{p_i}{n}}}{\left(\bar{X}_T^{\ast,-i,n}\right)^{p_i} \left(\bar{X}_T^{\ast,n}\right)^{p_i}} \left|X_T^{\ast,i,n}- X_T^{i,n}\right|  \max\left\{\left(X_T^{\ast,i,n}\right)^{\frac{p_i}{n}-1} , \left(X_T^{i,n}\right)^{\frac{p_i}{n}-1}\right\}.
\end{align*}
We only give the detailed proof on the event $\{X_T^{i,n} \leq X_T^{\ast,i,n}\}$, and skip similar arguments on the event $\{X_T^{i,n} > X_T^{\ast,i,n}\}$. It follows from $p_i/n\in(0,1)$ that
\begin{align*}
 &\left|\left(\bar{X}_T^{\ast,-i,n}\right)^{-\theta_i\gamma_i} -\left(\bar{X}_T^{\ast,n}\right)^{-\theta_i\gamma_i} \right|\nonumber\\
 &\qquad\leq \frac{p_i}{n}\frac{\left(\prod_{j\neq i}X_T^{\ast,j,n}\right)^{\frac{p_i}{n}}}{\left(\bar{X}_T^{\ast,-i,n}\right)^{p_i} \left(\bar{X}_T^{\ast,n}\right)^{p_i}} \left|X_T^{\ast,i,n}- X_T^{i,n}\right|  \max\left\{\left(X_T^{\ast,i,n}\right)^{\frac{p_i}{n}-1} , \left(X_T^{i,n}\right)^{\frac{p_i}{n}-1}\right\} \\
 &\qquad = \frac{p_i}{n}\frac{\left(\prod_{j\neq i}X_T^{\ast,j,n}\right)^{\frac{p_i}{n}} \left(X_T^{i,n}\right)^{\frac{p_i}{n}-1}  }{\left(\bar{X}_T^{\ast,-i,n}\right)^{p_i} \left(\bar{X}_T^{\ast,n}\right)^{p_i}} \left(X_T^{\ast,i,n}- X_T^{i,n}\right) \\
 &\qquad \leq \frac{p_i}{n}\frac{X_T^{\ast,i,n}}{\left(\bar{X}_T^{\ast,n}\right)^{p_i}X_T^{i,n}} .
\end{align*}
This yields that
\begin{align*}
&\Ex\left[\frac{1}{\gamma_i}\left(X_T^{i,n}\right)^{\gamma_i} \left(\bar{X}_T^{\ast,-i,n}\right)^{-\theta_i\gamma_i}\right]
 -\Ex\left[\frac{1}{\gamma_i}\left(X_T^{i,n}\right)^{\gamma_i} \left(\bar{X}_T^{\ast,n}\right)^{-\theta_i\gamma_i} \right]\\
 &\qquad \leq \Ex\left[ \frac{1}{\gamma_i}\left(X_T^{i,n}\right)^{\gamma_i} \frac{p_i}{n} \frac{1}{\left(\bar{X}_T^{\ast,n}\right)^{p_i} } \frac{X_T^{\ast,i,n}}{X_T^{i,n}} \right]= \frac{p_i}{\gamma_i}\frac{1}{n} \Ex\left[X_T^{\ast,i,n}\left(X_T^{i,n}\right)^{\gamma_i-1} \left(\bar{X}_T^{\ast,n}\right)^{-p_i} \right].
\end{align*}
For $j=1,\ldots,n$, let $q_j = 2n$. Then $\frac{1}{4} + \frac{1}{4}+ \sum_{j=1}^nq_j^{-1} =1$. We have from the H\"{o}lder's inequality and the estimate \eqref{eq:bound-X-i-ast} in Lemma \ref{lem:auxiliarylemmas} that
\begin{align}\label{eq:moment-bounded}
 &\Ex\left[ X_T^{\ast,i,n} \left(X_T^{i,n}\right)^{\gamma_i-1} \left(\prod_{j=1}^n X_T^{\ast,j,n}\right)^{-\frac{p_i}{n}} \right]\nonumber \\
 &\qquad\leq \Ex\left[\left(X_T^{*,i,n}\right)^{4} \right]^{\frac{1}{4}} \Ex\left[\left(X_T^{i,n}\right)^{4(\gamma_i-1)} \right]^{\frac{1}{4}}  \prod_{j=1}^{n} \Ex\left[\left(X_T^{*,j,n}\right)^{-\frac{q_jp_i}{n}} \right]^{\frac{1}{q_j}}   \nonumber\\
 &\qquad = \Ex\left[\left(X_T^{*,i,n}\right)^{4} \right]^{\frac{1}{4}} \Ex\left[\left(X_T^{i,n}\right)^{4(\gamma_i-1)} \right]^{\frac{1}{4}}  \prod_{j=1}^{n} \Ex\left[\left(X_T^{*,j,n}\right)^{-2p_i} \right]^{\frac{1}{2n}}     \nonumber \\
 &\qquad \leq \left\{D_4 D_{4(\gamma_i-1)}\right\}^{\frac{1}{4}}\prod_{j=1}^{n} \left\{(D_{-2p_i}\right\}^{\frac{1}{2n}}
  =  \left\{D_4 D_{4(\gamma_i-1)}\right\}^{\frac{1}{4}}\left\{D_{-2p_i}\right\}^{\frac{1}{2}}:=D,
\end{align}
where the constant $D>0$ is independent of $n$. Thus we have that
\begin{align}\label{eq:I-i-1}
  I_1^i &=  \Ex\left[ \frac{1}{\gamma_i} \left(X_T^{i,n}\right)^{\gamma_i} \left(\bar{X}_T^{\ast,-i,n}\right)^{-\theta_i\gamma_i} \right]
 -\Ex\left[ \frac{1}{\gamma_i}\left(X_T^{i,n}\right)^{\gamma_i}\left(\bar{X}_T^{\ast,n}\right)^{-\theta_i\gamma_i} \right] \leq \frac{1}{\gamma_i} \frac{p_iD}{n} = O\left(\frac{1}{n}\right).
\end{align}
Similarly, for the term $I_2^i$, we can apply H\"{o}lder inequality and the estimate \eqref{eq:lem-bar-X-m} in Lemma \ref{lem:auxiliarylemmas} to get that, for some constant $C>0$ independent of $n$,
\begin{align}\label{eq:I-i-2}
 & \Ex\left[ \frac{1}{\gamma_i}\left(X_T^{i,n}\right)^{\gamma_i} \left(\bar{X}_T^{\ast,n}\right)^{-\theta_i \gamma_i}\right]
 -  \Ex\left[ \frac{1}{\gamma_i}\left(X_T^{i,n}\right)^{\gamma_i} \left(m_T^*\right)^{-\theta_i \gamma_i } \right] \\
 &\qquad \leq \frac{1}{\gamma_i} \Ex\left[ \left(X_T^{i,n}\right)^{\gamma_i}
 \left(\left(\bar{X}_T^{\ast,n}\right)^{-p_i} - \left(m_T^*\right)^{-p_i}\right) \right] \nonumber\\
 &\qquad \leq \frac{1}{\gamma_i} \Ex\left[\left(X_T^{i,n}\right)^{\gamma_i } \frac{1}{ \left(m_T^*\right)^{p_i}\left(\bar{X}_T^{*,n}\right)^{p_i}} \left[ \left(m_T^*\right)^{p_i} - \left(\bar{X}_T^{\ast,n}\right)^{p_i} \right] \right] \nonumber \\
 &\qquad\leq \frac{1}{\gamma_i}\Ex\left[\left(X_T^{i,n}\right)^{\gamma_i} \frac{1}{\left(m_T^*\right)^{p_i}\left(\bar{X}_T^{\ast,n}\right)^{p_i}} p_i\left| m_T^* - \bar{X}_T^{*,n}\right|\max\left\{\left(m_T^*\right)^{p_i-1}, \left(\bar{X}_T^{*,n}\right)^{p_i-1}\right\}\right]\nonumber\\
 &\qquad\leq C \left\{\Ex\left[\left|m_T^* - \bar{X}_T^{*,n}\right|^2\right]\right\}^{\frac{1}{2}}= O\left(n^{-\frac{1}{4}}\right). \nonumber
\end{align}
For the second term of r.h.s. of \eqref{eq:J-J}, we can follow similar argument in proving \eqref{eq:I-i-2} to get that
\begin{equation}\label{eq:J-bar-J-1}
  \bar{J}_i^n(\pi^{*,i,n};m^*)- J_i(\pi^{*,i,n}, \pi^{*,-i,n}) = O\left(n^{-\frac{1}{4}}\right).
\end{equation}

We next claim that
\begin{align}\label{eq:Jn-bar-Jn}
 &\sup_{\pi^i\in \mathcal{A}^i}\bar{J}_i^n(\pi^i; m^*) - \bar{J}_i^n(\pi^{*,i,n};m^*) = O\left(n^{-\frac{1}{2}}\right),
\end{align}
where we recall the definition of $\pi^{*,i,n}=(\pi_t^{*,i,n})_{t\in[0,T]}$ in \eqref{eq:pi-epsilon-i}. It follows from \eqref{eq:pi-epsilon-i} that, $\Px$-a.s.
\begin{equation*}\label{eq:2}
 (\gamma_i-1)\left[\sigma_i^2 + (\sigma_i^0)^2\right]\pi_t^{*,i,n} -\theta_i\gamma_i\sigma_i^0\sigma^0\pi_t^*
  - \Lambda_t^{f,i,n}(1-\pi_t^{*,i,n})^{\gamma_i-1} + \Lambda_t^{f,i,n} + b_i = 0 .
\end{equation*}
Recall that the best response solution $\tilde{\pi}^{i,n}=\tilde{\pi}^{i,n}(t,{\bm\lambda})$ of the auxiliary control problem ($\bf{P_n}$) satisfies
\eqref{eq:FOC-22-PP0} in Lemma~\ref{lem:optimum-Pn}. Let $\tilde{\pi}_t^{i,n}=\tilde{\pi}^{i,n}(t,{\bm\Lambda}_t^f)$ and it holds that, $\Px$-a.s.
\begin{equation}\label{eq:Phi-i-value}
  \Phi_i(t,\Lambda_t^{f,i,n},\tilde{\pi}_t^{i,n})  = O\left( n^{-1}\right),\qquad \Phi_i(t,\Lambda_t^{f,i,n},\pi_t^{*,i,n}) = 0.
\end{equation}
For any $\xi\in [\min\{\tilde{\pi}_t^{i,n}, \pi_t^{*,i,n}\}, \max\{\tilde{\pi}_t^{i,n}, \pi_t^{*,i,n}\}]$, by Lemma~\ref{lem:optimum-Pn}, we have $\Px$-a.s. that $\partial_{\pi}\Phi_i(t,\Lambda_t^{f,i,n},\xi)\in[C_1,C_2]$ for constants $-\infty<C_1<C_2<\infty$, which are independent of $(t,{\bm\lambda},i,n)$. Using \eqref{eq:Phi-i-value} and the mean value theorem, we arrive at
\begin{align}\label{eq:pio1overn}
 \left|\tilde{\pi}_t^{i,n} - \pi_t^{*,i,n} \right| = O\left(\frac{1}{n}\right),\quad \forall t\in[0,T].
 \end{align}
Building upon the estimate \eqref{eq:pio1overn}, we next focus on the proof of \eqref{eq:Jn-bar-Jn}. In fact, by virtue of \eqref{eq:wealth-X-i-P-n}, it follows from It\^o's formula that, for any admissible strategy $\pi^i=(\pi_t^i)_{t\in[0,T]}\in{\cal A}^i$,
\begin{align}\label{eq:X-i-gamma}
  \frac{d(X_t^{i,n})^{\gamma_i}}{(X_{t-}^{i,n})^{\gamma_i}}
  & = \gamma_i [r+ (b_i+ \Lambda_t^{f,i} )\pi_t^i + \frac{1}{2}(\gamma_i-1)(\sigma_i^2(\pi_t^i)^2+ (\sigma_i^0 \pi_t^i)^2)]dt + \gamma_i \pi_t^i \sigma_i dW_t^i \nonumber\\
  & \quad  + \gamma_i \pi_t^i \sigma_i^0 dW_t^0  +[(1-\pi_{t}^i)^{\gamma_i} -1]dN_t^{i,n} \nonumber\\
  &:= F(t,\pi_t^i)dt + \gamma_i \pi_t^i \sigma_i dW_t^i  + \gamma_i \pi_t^i \sigma_i^0 dW_t^0  +[(1-\pi_{t}^i)^{\gamma_i} -1]dN_t^{i,n},
\end{align}
with $F(t,\pi_t^i) := \gamma_i \{r+ (b_i+\Lambda_t^{f,i,n})\pi_t^i + \frac{1}{2}(\gamma_i-1)(\sigma_i^2+ (\sigma_i^0 )^2) (\pi_t^i)^2\}$. This is equivalent to
\begin{align*}
 (X_t^{i,n})^{\gamma_i} & = x_i^{\gamma_i}\exp\Bigg\{\int_{0}^{T}F(t,\pi_t^i)dt + \int_{0}^{T} \gamma_i \pi_t^i \sigma_i dW_t^i  + \int_{0}^{T} \gamma_i \pi_t^i \sigma_i^0 dW_t^0 + \int_{0}^{T}\log(1-\pi_{t}^i)^{\gamma_i} dN_t^{i,n} \Bigg\}.
 \end{align*}
Let us recall that, for $t\in[0,T]$,
\begin{align*}
  m_t^* &= x_0 \exp\left\{\int_{0}^{t}\left(\eta(s;\pi_s^*) -\frac{1}{2}(\sigma^0\pi_s^{*})^2\right)ds + \int_{0}^{t}\sigma^0\pi_s^{*} dW_s^0\right\}.
\end{align*}
Therefore, for any $\pi^i\in{\cal A}^i$, it holds that
\begin{align}\label{eq:expectation-objective}
 & \Ex\left[\frac{1}{\gamma_i}\left(X_T^{i,n}\right)^{\gamma_i} \left(m_T^*\right)^{-\theta_i\gamma_i}\right]\nonumber  \\
 &\quad = \frac{1}{\gamma_i}\Ex\Bigg[ x_i^{\gamma_i} x_0^{-\theta_i\gamma_i}\exp\Bigg\{\int_{0}^{T}F(t,\pi_t^i)dt
 + \int_{0}^{T}\gamma_i\pi_t^i\sigma_i dW_t^i + \int_{0}^{T}\gamma_i\pi_t^i\sigma_i^0dW_t^0     \nonumber\\
 &\qquad  + \int_{0}^{T}\log (1-\pi_{t}^i)^{\gamma_i} dN_t^{i,n} - \int_{0}^{T}\theta_i\gamma_i(\eta(t;\pi_t^*) -\frac{1}{2}(\sigma^0\pi_t^{*})^2)dt - \int_{0}^{T}\theta_i\gamma_i \sigma^0\pi_t^{*}dW_t^0 \Bigg\}\Bigg] \nonumber\\
&\quad = \frac{1}{\gamma_i}x_i^{\gamma_i}x_0^{-\theta_i\gamma_i}\Ex\Bigg[\exp\Bigg\{\int_{0}^{T} G(t,\pi_t^i,\pi_t^*)dt + \int_{0}^{T}\gamma_i \pi_t^i\sigma_idW_t^i + \int_{0}^{T}\gamma_i(\pi_t^i\sigma_i^0 -\theta_i\sigma^0\pi_t^{*}) dW_t^0  \nonumber\\
 &\qquad\quad + \int_{0}^{T}\log (1-\pi_{t}^i)^{\gamma_i}dN_t^{i,n} \Bigg\}\Bigg]
\end{align}
with $G(t,\pi_t^i,\pi_t^*):=  F(t,\pi_t^i) - \theta_i\gamma_i(\eta(t;\pi_t^*) -\frac{1}{2}(\sigma^0\pi_t^{*})^2)$.

Let $\tilde{X}^{i,n}= (\tilde{X}_t^{i,n})_{t\in[0,T]}$ and $X^{*,i,n} = (X_t^{*,i,n})_{t\in[0,T]}$ be the wealth processes corresponding to the strategies $\tilde{\pi}^{i,n}$ and $\pi^{*,i,n}$, respectively. Using the representation \eqref{eq:expectation-objective} with $\pi^i$ replaced by $\tilde{\pi}^{i,n}$ and $\pi^{*,i,n}$ respectively, we deduce that
\begin{align}\label{eq:J-bar-J-2}
  & \sup_{\pi^i \in \mathcal{A}^i}\bar{J}_i^n(\pi^{i}; m^* ) - \bar{J}_i^n(\pi^{\ast,i,n};m^*) = \Ex\left[\frac{1}{\gamma_i}(\tilde{X}_T^{i,n})^{\gamma_i} (m_T^*)^{-\theta_i\gamma_i}\right] -
    \Ex\left[\frac{1}{\gamma_i}(X_T^{*,i,n})^{\gamma_i} (m_T^*)^{-\theta_i\gamma_i}\right] \nonumber \\
 &\quad = \frac{1}{\gamma_i} x_i^{\gamma_i}x_0^{-\theta_i\gamma_i}\Ex\Bigg[\exp\Bigg\{\int_{0}^{T} G(t,\tilde{\pi}_t^{i,n};\pi_t^*)dt + \int_{0}^{T}\gamma_i\tilde{\pi}_t^{i,n}\sigma_idW_t^i  + \int_{0}^{T}\gamma_i\big(\tilde{\pi}_t^{i,n}\sigma_i^0-\theta_i\sigma^0\pi_t^{*}\big)dW_t^0   \nonumber\\
 &\qquad + \int_{0}^{T}\log (1-\tilde{\pi}_t^{i,n})^{\gamma_i} dN_t^{i,n}  \Bigg\}- \exp\Bigg\{\int_{0}^{T}G(t,\pi_t^{*,i,n},\pi_t^*)dt
  + \int_{0}^{T}\gamma_i\pi_t^{*,i,n}\sigma_i dW_t^i  \nonumber\\
 &\qquad + \int_{0}^{T}\gamma_i(\pi_t^{*,i,n}\sigma_i^0 -\theta_i\sigma^0\pi_t^{*})dW_t^0
   + \int_{0}^{T}\log(1-\pi_t^{*,i,n})^{\gamma_i}dN_t^{i,n} \Bigg\}\Bigg]:= I_3^i.
\end{align}
After a straightforward calculation, the term $I_3^i$ can be further rewritten as:
\begin{align}\label{eq:I-3-i}
  I_3^i & = \frac{1}{\gamma_i} x_i^{\gamma_i}x_0^{-\theta_i\gamma_i}\Ex\Bigg[\Bigg\{ \exp\Bigg(\int_{0}^{T}(F(t,\tilde{\pi}_t^{i,n})- F(t,\pi_t^{*,i,n})) dt
   + \int_{0}^{T}\gamma_i \sigma_i(\tilde{\pi}_t^{i,n}-\pi_t^{*,i,n}) dW_t^i \nonumber\\
 & \quad  + \int_{0}^{T}\gamma_i\sigma_i^0(\tilde{\pi}_t^{i,n}-\pi_t^{*,i,n})dW_t^0  + \int_{0}^{T}\log\frac{(1-\tilde{\pi}_t^{i,n})^{\gamma_i}}{(1-\pi_t^{*,i,n})^{\gamma_i}} dN_t^{i,n} \Bigg) -1 \Bigg\} \nonumber \\
 &\quad\times \exp\Bigg\{\int_{0}^{T}G(t,\pi_t^{*,i,n},\pi_t^*)dt + \int_{0}^{T}\gamma_i\pi_t^{*,i,n}\sigma_i dW_t^i + \int_{0}^{T}\gamma_i(\pi_t^{*,i,n}\sigma_i^0 -\theta_i\sigma^0\pi_t^{*})dW_t^0 \nonumber\\
 &\qquad\qquad + \int_{0}^{T}\log(1-\pi_t^{*,i,n})^{\gamma_i}dN_t^{i,n} \Bigg\}\Bigg]  \nonumber \\
 & = \frac{1}{\gamma_i} \Ex\Bigg[\Bigg\{\exp\Bigg(\int_{0}^{T}(F(t,\tilde{\pi}_t^{i,n})- F(t,\pi_t^{*,i,n})) dt
  +\int_{0}^{T}\gamma_i(\tilde{\pi}_t^{i,n}-\pi_t^{*,i,n})(\sigma_i dW_t^i +\sigma_i^0dW_t^0) \nonumber\\
 &\qquad\quad   + \int_{0}^{T}\log\frac{(1-\tilde{\pi}_t^{i,n})^{\gamma_i}}{(1-\pi_t^{*,i,n})^{\gamma_i}} dN_t^{i,n} \Bigg) -1 \Bigg\}(X_T^{*,i,n})^{\gamma_i} \left(m_T^*\right)^{-\theta_i\gamma_i} \Bigg].
\end{align}
Applying the Cauchy-Schwarz inequality to \eqref{eq:I-3-i} results in
\begin{align}\label{eq:I-3-i-Cauchy}
  I_3^i & \leq \frac{1}{\gamma_i}\Ex\left[(X_T^{*,i,n})^{2\gamma_i}(m_T^*)^{-2\theta_i\gamma_i}\right]^{\frac{1}{2}} \Ex\Bigg[\Bigg|\exp\Bigg\{\int_{0}^{T}(F(t,\tilde{\pi}_t^{i,n})- F(t,\pi_t^{*,i,n})) dt\nonumber \\
 & \qquad  + \int_{0}^{T}\gamma_i(\tilde{\pi}_t^{i,n}-\pi_t^{*,i,n})(\sigma_idW_t^i +\sigma_i^0dW_t^0)  + \int_{0}^{T}\log\frac{(1-\tilde{\pi}_t^{i,n})^{\gamma_i}}{(1-\pi_t^{*,i,n})^{\gamma_i}} dN_t^{i,n} \Bigg\} -1 \Bigg|^2\Bigg]^{\frac{1}{2}} \nonumber\\
 & = \frac{1}{\gamma_i}\Ex\left[(X_T^{*,i,n})^{2\gamma_i}(m_T^*)^{-2\theta_i\gamma_i}\right]^{\frac{1}{2}} \Ex\Bigg[ \exp\Bigg\{ 2\int_{0}^{T}(F(t,\tilde{\pi}_t^{i,n})- F(t,\pi_t^{*,i,n}))dt \nonumber\\
 & \qquad + 2\int_{0}^{T}\gamma_i(\tilde{\pi}_t^{i,n}-\pi_t^{*,i,n})(\sigma_idW_t^i+\sigma_i^0dW_t^0)   + 2\int_{0}^{T}\log\frac{(1-\tilde{\pi}_t^{i,n})^{\gamma_i}}{(1-\pi_t^{*,i,n})^{\gamma_i}} dN_t^{i,n} \Bigg\}\nonumber \\
 &\qquad - 2\exp\Bigg\{ \int_{0}^{T}(F(t,\tilde{\pi}_t^{i,n})- F(t,\pi_t^{*,i,n})) dt  + \int_{0}^{T}\gamma_i(\tilde{\pi}_t^{i,n}-\pi_t^{*,i,n})(\sigma_idW_t^i +\sigma_i^0dW_t^0)  \nonumber\\
 & \qquad   + \int_{0}^{T}\log\frac{(1-\tilde{\pi}_t^{i,n})^{\gamma_i}}{(1-\pi_t^{*,i,n})^{\gamma_i}} dN_t^{i,n} \Bigg\} + 1 \Bigg]^{\frac{1}{2}}.
\end{align}
It follows from \eqref{eq:bound-X-i-ast} in Lemma \ref{lem:auxiliarylemmas} that $\Ex[(X_T^{*,i,n})^{2\gamma_i}(m_T^*)^{-2\theta_i\gamma_i}]$ is bounded and the bound is independent of $(i,n)$. Then, by \eqref{eq:I-3-i-Cauchy}, in order to verify \eqref{eq:Jn-bar-Jn}, we need to show that
\begin{align}\label{eq:estimateo1n}
 & \Ex\Bigg[\exp\Bigg\{\int_{0}^{T}(F(t,\tilde{\pi}_t^{i,n})- F(t,\pi_t^{*,i,n})) dt
   + \int_{0}^{T}\gamma_i(\tilde{\pi}_t^{i,n}-\pi_t^{*,i,n})(\sigma_idW_t^i +\sigma_i^0dW_t^0)     \nonumber\\
 & \qquad + \int_{0}^{T}\log\frac{(1-\tilde{\pi}_t^{i,n})^{\gamma_i}}{(1-\pi_t^{*,i,n})^{\gamma_i}} dN_t^{i,n} \Bigg\} \Bigg] =1+ O\left(\frac{1}{n}\right).
\end{align}
We introduce the density process $L=(L_t)_{t\in[0,T]}$ satisfying the following SDE under the original probability measure $\mathbb{P}$ that
\begin{align}\label{Girsanov-L}
  \frac{dL_t}{L_{t-}}&  = \gamma_i \sigma_i (\tilde{\pi}_t^{i,n}-\pi_t^{*,i,n}) dW_t^i
  + \gamma_i\sigma_i^0(\tilde{\pi}_t^{i,n}-\pi_t^{*,i,n})dW_t^0  + \left\{\frac{(1-\tilde{\pi}_t^{i,n})^{\gamma_i}}{(1-\pi_t^{*,i,n})^{\gamma_i}} -1 \right\} dM_t^{f,i,n}.
\end{align}
In view that $\tilde{\pi}^{i,n}, \pi^{*,i,n}\in{\cal A}^i$, the density process $L$ is in fact a martingale. We then define a probability measure $\mathbb{Q}\sim\Px$ by
\begin{align}\label{eq:Q}
  \frac{d \mathbb{Q}}{d\mathbb{P}}\big|_{\mathcal{G}_t} = L_t,\quad t\in[0,T].
\end{align}
Using the change of measure and \eqref{eq:pio1overn}, we obtain the existence of a constant $C>0$ independent of $n$ such that
\begin{align*}
 & \Ex\Bigg[\exp\Bigg\{\int_{0}^{T}(F(t,\tilde{\pi}_t^{i,n})- F(t,\pi_t^{*,i,n})) dt+ \int_{0}^{T}\gamma_i \sigma_i(\tilde{\pi}_t^{i,n}-\pi_t^{*,i,n})dW_t^i     \nonumber\\
 & \qquad\quad + \int_{0}^{T}\gamma_i\sigma_i^0(\tilde{\pi}_t^{i,n}-\pi_t^{*,i,n})dW_t^0 + \int_{0}^{T}\log\frac{(1-\tilde{\pi}_t^{i,n})^{\gamma_i}}{(1-\pi_t^{*,i,n})^{\gamma_i}} dN_t^{i,n} \Bigg\}\Bigg] \\
 &\quad = \Ex^{\mathbb{Q}}\Bigg[\exp\Bigg\{\int_{0}^{T}(F(t,\tilde{\pi}_t^{i,n})- F(t,\pi_t^{*,i,n})) dt
   + \int_{0}^{T}\frac{1}{2}\gamma_i^2 (\sigma_i^2+(\sigma_i^0)^2) (\tilde{\pi}_t^{i,n}-\pi_t^{*,i,n})^2 dt      \nonumber\\
 &\qquad\qquad\quad+\int_{0}^{T}\left(\frac{(1-\tilde{\pi}_t^{i,n})^{\gamma_i}}{(1-\pi_t^{*,i,n})^{\gamma_i}}-1\right)
 \Lambda_t^{f,i,n}dt \Bigg\}\Bigg] \nonumber\\
 &\quad =\Ex^{\mathbb{Q}}\Bigg[\exp\Bigg\{\int_{0}^{T} \bigg[\gamma_i(b_i+\Lambda_t^{f,i,n})(\tilde{\pi}_t^{i,n}-\pi_t^{*,i,n}) +\frac{1}{2}\gamma_i(\gamma_i-1)(\sigma_i^2+(\sigma_i^0)^2)((\tilde{\pi}_t^{i,n})^2-(\pi_t^{*,i,n})^2)     \nonumber\\
 &\qquad \qquad + \frac{1}{2}\gamma_i^2 (\sigma_i^2+(\sigma_i^0)^2) (\tilde{\pi}_t^{i,n}-\pi_t^{*,i,n})^2 + \left(\frac{(1-\tilde{\pi}_t^{i,n})^{\gamma_i}}{(1-\pi_t^{*,i,n})^{\gamma_i}}-1\right)\Lambda_t^{f,i,n}\bigg]dt \Bigg\}\Bigg]\\
 &\quad \leq \Ex^{\mathbb{Q}}\left[\exp\left(C \int_{0}^{T}|\tilde{\pi}_t^{i,n}-\pi_t^{*,i,n}| dt \right)\right]\\
 &\quad = 1+ O\left(\frac{1}{n}\right).
\end{align*}
Here, $\Ex^{\Qx}$ denotes the expectation operator under $\Qx$. This shows the validity of \eqref{eq:estimateo1n}. It then follows from \eqref{eq:estimateo1n} and \eqref{eq:I-3-i-Cauchy} that \eqref{eq:Jn-bar-Jn} holds. Combining \eqref{eq:I-i-1}-\eqref{eq:Jn-bar-Jn}, we can conclude the desired estimation that $\sup_{\pi^i \in \mathcal{A}^i} J_i(\pi^{i},\pi^{*,-i,n}) - J_i(\pi^{*,i,n}, \pi^{*,-i,n}) \leq O(n^{-\frac{1}{4}})$. Thus, we complete the proof of the theorem.
\end{proof}

\section{Conclusions}\label{sec:con}
This paper revisits the MFG and the $n$-player game under CRRA relative performance by allowing risky assets to have contagious jumps, which are modelled by a multi-dimensional mutually exciting nonlinear Hawkes process. As a first attempt to such problems to accommodate controlled jumps, it is assumed for tractability in the present paper that the limiting model has constant parameters. By using the FBSDE and stochastic maximum principle arguments, a deterministic Nash equilibrium for the MFG can be characterized as a function of the deterministic limiting intensity process. Furthermore, using the information of the MFE, we are able to construct a good approximation of the Nash equilibrium for the large but finite population game and the order of the approximation error is explicitly obtained.

Based on our current study, some future research directions can be considered. First, it will be attractive to consider both $n$-agent model and the limiting model with general random parameters. A deterministic mean field equilibrium may no longer exist. The existence of a mean field equilibria and the verification of an approximate Nash equilibrium for the $n$-player game will require different mathematical arguments. Second, our work may pave the way to consider other sophisticated default intensity processes. For example, the default intensity of each risky asset may depend on the asset price itself or other stochastic factors. Some novel analysis for the mean field FBSDE with jumps are in demand to tackle the MFG problem.

\ \\
\noindent
\textbf{Acknowledgements}\ We sincerely thank two anonymous referees for their helpful comments on the presentation of this paper. L. Bo is supported by Natural Science Basic Research Program of Shaanxi  (Program No. 2023-JC-JQ-05) and National Natural Science Foundation of China (Grant No. 11971368). S. Wang is supported by the Fundamental Research Funds for the Central Universities (Grant No. WK3470000024). X. Yu is supported by the Hong Kong Polytechnic University research (Grant No. P0031417 and No. P0039251).




\begin{appendix}

\section{Proofs of Some Auxiliary Results}\label{sec:proofs}

\begin{proof}[Proof of  Lemma \ref{lem:Phi}]
 Recall the definition of $ \Phi(\pi,\lambda): (-\infty, 1)\times \R_{+} \rightarrow \R $ in \eqref{eq:Phi}. As $0<\gamma <1$, the partial derivative is given by, for $(\pi,\lambda)\in(-\infty, 1)\times \R_{+}$,
\begin{equation*}
  \partial_{\pi} \Phi(\pi,\lambda) = (\gamma-1)[ \sigma^2 + (\sigma^0)^2]  -\theta\gamma(\sigma^0)^2 + (\gamma-1)\lambda(1-\pi)^{\gamma-2}  < 0,
\end{equation*}
which implies that $\Phi(\pi,\lambda)$ is monotonically decreasing with respect to $\pi\in(-\infty,1)$. Note that $ \Phi(0,\lambda) = b >0 $ and $ \lim_{\pi\uparrow1}\Phi(\pi,\lambda)=-\infty$. It is deduced that there exists a unique $ \pi^* \in (0,1-\epsilon_0]$ such that $\Phi(\pi^* ,\lambda)=0$, where the constant $\epsilon_0\in(0,1)$ is small enough. Moreover, as $ \partial_{\pi} \Phi(\pi,\lambda)  < 0 $ for all $(\pi,\lambda)\in(-\infty,1)\times\R_+$,
it follows from the implicit function theorem that, there exists a unique continuous function $\phi$ such that $\pi^* = \phi(\lambda)$, and $\phi$ has a continuous partial derivative with respect to $\lambda$. We also note that
\begin{equation*}
 \partial_{\lambda} \phi(\lambda) = -\frac{\partial_{\lambda} \Phi(\pi^*,\lambda)}{\partial_{\pi} \Phi(\pi^*,\lambda)}
  = - \frac{1- (1-\pi^*)^{\gamma -1}}{(\gamma-1)[\sigma^2 + (\sigma^0)^2]-\theta\gamma(\sigma^0)^2 + (\gamma-1)\lambda(1-\pi^*)^{\gamma-2} } < 0,
\end{equation*}
and $ |\partial_{\lambda} \phi(\lambda) |\leq K$ for some constant $K$ independent of $\lambda$. The desired result follows that the function $\phi$ is decreasing in $\lambda$ and is Lipschitz continuous with respect to $\lambda$.
\end{proof}

\begin{proof}[Proof of Lemma \ref{quanprop-1}]
Note that the mean filed  equilibrium $t\mapsto\pi^*_t$ in Theorem \ref{thm:optimal-MFG} is a positive deterministic function that $\pi_t^* = \phi(\lambda_t^{f,o})\in (0, 1-\epsilon_0)$ for $0<\gamma<1$. For each fixed $t\in[0,T]$, by straightforward computations, we can derive the derivatives that
\begin{align*}
\partial_b \phi &= -\frac{\partial_b \Phi}{\partial_{\pi}\Phi}
 = -\frac{1}{(\gamma-1)[\sigma^2+(\sigma^0)^2] -\theta\gamma(\sigma^0)^2 +(\gamma-1)\lambda_t^o(1-\pi_t^*)^{\gamma-2}}>0,\\
\partial_{\sigma} \phi & = -\frac{\partial_{\sigma} \Phi}{\partial_{\pi} \Phi}
 = -\frac{2(\gamma-1)\sigma\pi_t^*}{(\gamma-1)[\sigma^2+(\sigma^0)^2] -\theta\gamma(\sigma^0)^2 +(\gamma-1)\lambda_t^o(1-\pi_t^*)^{\gamma-2}} <0,\\
\partial_{\sigma_0} \phi &= -\frac{ \partial_{\sigma_0} \Phi}{\partial_{\pi} \Phi}
 = -\frac{2(\gamma-1)\sigma_0\pi_t^* -2\theta\gamma\sigma_0\pi_t^*}{(\gamma-1)[\sigma^2+(\sigma^0)^2] -\theta\gamma(\sigma^0)^2 +(\gamma-1)\lambda_t^o(1-\pi_t^*)^{\gamma-2}} <0,\\
\partial_{\gamma} \phi &= - \frac{\partial_{\gamma} \phi}{\partial_{\pi} \Phi}
 = - \frac{\sigma^2 +(1-\theta)(\sigma^0)^2\pi_t^* - \lambda_t^o\log(1-\pi_t^*)(1-\pi_t^*)^{-\frac{1}{\gamma}}} {(\gamma-1)[\sigma^2 + (\sigma^0)^2] -\theta\gamma(\sigma^0)^2 + (\gamma-1)\lambda_t^o(1-\pi_t^*)^{\gamma-2}}<0,\\
\partial_{\theta} \phi &= - \frac{\partial_{\theta} \Phi}{\partial_{\pi} \Phi}
= - \frac{-\gamma(\sigma^0)^2\pi_t^*}{(\gamma-1)[\sigma^2 + (\sigma^0)^2] -\theta\gamma(\sigma^0)^2 + (\gamma-1)\lambda_t^o(1-\pi_t^*)^{\gamma-2}} < 0.
\end{align*}
Thus, the claimed monotonicity results follow directly.
\end{proof}

\begin{proof}[Proof of Lemma~\ref{lem:optimum-Pn}]
Let ${\cal A}_t^i$ be the admissible control set starting with any time $t\in[0,T]$. Then, we can define the value function of the auxiliary control problem $(\bf{P_n})$ given by, for $(t,x,m,\bm{\lambda})\in [0,T]\times \mathbb{R}_{+}\times\R_+\times \Xi_{\Lambda}^n$,
\begin{equation}\label{eq:value-function-P}
  V(t,x,m,\bm{\lambda}):= \sup_{\pi^i \in {\cal A}_t^i} \Ex\left[\frac{1}{\gamma_i}(X_T^{i,n})^{\gamma_i}(m_T^*)^{-\theta_i\gamma_i}\Big|X_t^{i,n} = x, m_t^* =m,\bm{\lambda}_t = \bm{\lambda}  \right],
\end{equation}
where $\bm{\lambda}_t:=(\lambda_t^{1},\ldots,\lambda_t^{n})^{\top}$ for $t\in[0,T]$. The value function \eqref{eq:value-function-P} is then associated with the HJB equation that
\begin{align*}
 0& = \partial_t V(t,x,m,\bm{\lambda})+ \sum_{j=1}^n \partial_{\lambda_j} V(t,x,m,\bm{\lambda})\alpha_j(\lambda_{\infty}^j -\lambda_j)\nonumber\\
 &\quad  + \sum_{j\neq i}^{n}f(\lambda_j) \left(V\left(t,x,m,\bm{\lambda} + \frac{\beta_i\varsigma_j}{n}\bm{e}_j^n\right) - V(t,x,m,\bm{\lambda})\right) \nonumber \\
 &\quad + \sup_{\pi^i\in(-\infty,1)}\Bigg\{\partial_x V(t,x,m,\bm{\lambda}) x(r+ (b_i+f(\lambda_i))\pi^i) + \partial_m V(t,x,m,\bm{\lambda})\eta(t;\pi_t^*)m  \nonumber\\
 &\qquad\qquad\qquad\quad + \frac{1}{2}\partial_{xx}V(t,x,m,\bm{\lambda})x^2(\sigma_i^2+ (\sigma_i^0)^2)(\pi^i)^2 + \frac{1}{2}\partial_{mm}V(t,x,m,\bm{\lambda}) m^2(\sigma^0\pi_t^{\ast})^2 \nonumber\\
  &\quad+ \partial_{xm}V(t,x,m,\bm{\lambda}) xm\sigma_i^0\pi^i \sigma^0\pi_t^{\ast}
    + f(\lambda_i)\left(V\left(t,(1-\pi^i)x,m,\bm{\lambda}+ \frac{\beta_i\varsigma_i}{n}\bm{e}_i^n\right) - V(t,x,m,\bm{\lambda})\right) \Bigg\}
\end{align*}
with the terminal condition $V(T,x,m,\bm{\lambda}) = \frac{1}{\gamma_i}x^{ \gamma_i} m^{-\theta_i\gamma_i}$. Here, ${\bm e}_i^n$ denotes the $n$-dimensional column vector whose $i$-th entry is $1$ and remaining ones are $0$.  Let us consider the decoupled form that $V(t,x,m,\bm{\lambda}) = \frac{1}{\gamma_i}x^{\gamma_i}m^{-\theta_i\gamma_i}B(t,\bm{\lambda})$, where $B(t,\bm{\lambda})$ solves the following equation:
\begin{align}\label{eq:HJB-B}
0 &=\partial_t B(t,\bm{\lambda}) + \frac{1}{2}\theta_i\gamma_i(\theta_i\gamma_i +1) (\sigma^0\pi_t^{*})^2B(t,\bm{\lambda})  -\theta_i\gamma_i\eta(t;\pi_t^*)B(t,\bm{\lambda}) + \sup_{\pi^i\in(-\infty,1)}{\cal H}(t,{\bm\lambda};\pi^i)\nonumber \\
 &\quad + \sum_{j=1}^{n}\partial_{\lambda_j} B(t,\bm{\lambda}) \alpha_j(\lambda_{\infty}^j - \lambda_j)+  \sum_{j\neq i}f(\lambda_j) \left(B\left(t,\bm{\lambda}+\frac{\beta_i\varsigma_j}{n}\bm{e}_j^n\right)  - B(t,\bm{\lambda}) \right),
\end{align}
where the terminal condition is given by $B(T,\bm{\lambda})=1$, and ${\cal H}(t,{\bm\lambda};\pi^i)$ corresponds to the Hamiltonian operator that
\begin{align}\label{eq:Hamilton-P}
{\cal H}(t,{\bm\lambda};\pi^i)&:= \gamma_i [ r+ (b_i + f(\lambda_i) )\pi^i ]B(t,{\bm\lambda})
  +\frac{1}{2}\gamma_i(\gamma_i-1)(\sigma_i^2+ (\sigma_i^0)^2)(\pi^i)^2 B(t,{\bm\lambda})\nonumber \\
 &\quad -\theta_i \gamma_i^2\sigma_i^0\pi^i\sigma^0\pi_t^* B(t,{\bm\lambda} )  +  f(\lambda_i) \left((1-\pi^i)^{\gamma_i}B\left(t,{\bm\lambda} +\frac{\beta_i\varsigma_i}{n}{\bm e}_i^n\right) -B(t,{\bm\lambda})\right).
\end{align}
It follows from Theorem 4.1 in \cite{bocapponichen2019} and Proposition 4.3 in \cite{delong2008optimal} that \eqref{eq:HJB-B} admits a unique (positive) classical solution. By applying the first-order condition to ${\cal H}(t,{\bm\lambda};\pi^i)$ with respect to $\pi^i\in(-\infty,1)$, we obtain that the optimum $\tilde{\pi}^{i,n}=\tilde{\pi}^{i,n}(t,{\bm\lambda})\in(-\infty,1)$ in Eq.~\eqref{eq:HJB-B} satisfies
\begin{align}\label{eq:FOC-2-P}
b_i + f(\lambda_i) + (\gamma_i-1)(\sigma_i^2+ (\sigma_i^0)^2)\tilde{\pi}^{i,n} -\theta_i\gamma_i\sigma_i^0\sigma^0\pi_t^* - (1-\tilde{\pi}^{i,n})^{\gamma_i-1} \frac{B\left(t,\bm{\lambda}+\frac{\beta_i\varsigma_i}{n}\bm{e}_i^n\right)}{B(t,\bm{\lambda})}f(\lambda_i) = 0.
\end{align}
By the assumption $\bm{(A_O)}$, it holds that
\begin{align*}
\left|\frac{B\left(t,\bm{\lambda}+\frac{\beta_i\varsigma_i}{n}\bm{e}_i^n\right)}{B(t,\bm{\lambda})}-1\right|
=\frac{\left|B\left(t,\bm{\lambda}+\frac{\beta_i\varsigma_i}{n}\bm{e}_i^n\right)-B(t,\bm{\lambda})\right|}{B(t,\bm{\lambda})}\leq \frac{\frac{\beta_i\varsigma_i}{n}\left\|\partial_{\lambda_i}B(\cdot,\cdot)\right\|_{\infty}}{B(t,\bm{\lambda})}=O\left(\frac{1}{n}\right).
\end{align*}
This yields the desired result \eqref{eq:FOC-22-PP0}.

It follows from the assumption $\bm{(A_O)}$ and \eqref{eq:FOC-22-PP0} that, there exist a constant $K>0$ independent of $(i,n)$ such that $|\Phi_i(t,f(\lambda_i),\tilde{\pi}^{i,n})|\leq K$. Note that $\pi\mapsto\Phi_i(t,f(\lambda_i),\pi)$ is continuous and decreasing, $\lim_{\pi\to1} \Phi_i(t,f(\lambda_i),\pi) = -\infty $ and $\lim_{\pi\to -\infty}\Phi_i(t,f(\lambda_i),\pi)=+\infty$. We get the existence of some $(D_i,\epsilon_i)\in\R\times(0,1)$ independent of $(t,{\bm\lambda},n)$ such that $\tilde{\pi}^{i,n} \in [D_i, 1-\epsilon_i]$. We next prove that there exists a pair $(\tilde{D},\tilde{\epsilon})\in\R\times(0,1)$, which is independent of $(t,{\bm\lambda},i,n)$ such that $D_i\geq\tilde{D}$ and $\epsilon_i\geq\tilde{\epsilon}$. In fact, by the monotonicity of $\pi\mapsto\Phi_i(t,\lambda_i,\pi)$, it follows that
\begin{align*}
\Phi(t,f(\lambda_i),D_i) &= (\gamma_i-1)[\sigma_i^2+(\sigma_i^0)^2]D_i - \theta_i\gamma_i\sigma_i^0\sigma^0 \pi_t^* - f(\lambda_i)(1-D_i)^{\gamma_i-1} + f(\lambda_i) + b_i  \leq K,\nonumber\\
 \Phi(t,f(\lambda_i), 1-\epsilon_i)&= (\gamma_i-1)[\sigma_i^2 + (\sigma_i^0)^2](1- \epsilon_i) -  \theta_i\gamma_i\sigma_i^0\sigma^0 \pi_t^* - f(\lambda_i)(\epsilon_i)^{\gamma_i -1} + f(\lambda_i) + b_i \geq -K.
\end{align*}
Thanks to the assumption $\bm{(A_O)}$, we have that
\begin{align*}
  (\gamma_i-1)[\sigma_i^2+(\sigma_i^0)^2] D_i & \leq K + \theta_i\gamma_i\sigma_i^0\sigma^0 \pi_t^* +f(\lambda_i)(1-D_i)^{\gamma_i-1} - f(\lambda_i) - b_i
   \leq K + \sup_{i\geq 1}\theta_i\gamma_i\sigma_i^0\sigma^0.
\end{align*}
This yields that
\begin{align*}
 -D_i &\leq \frac{K + \sup_{i\geq 1}\theta_i\gamma_i\sigma_i^0\sigma^0}{(1-\gamma_i)[\sigma_i^2+(\sigma_i^0)^2]}\leq \frac{K + \sup_{i\geq 1}\theta_i\gamma_i\sigma_i^0\sigma^0}{\inf_{i\geq 1} (1-\gamma_i)[\sigma_i^2+(\sigma_i^0)^2]}: = -  \tilde{D} .
\end{align*}
Similarly, we also have that
\begin{align*}
  f(\lambda_i)(\epsilon_i)^{\gamma_i -1} & \leq K +(\gamma_i-1)[\sigma_i^2 + (\sigma_i^0)^2](1- \epsilon_i) -  \theta_i\gamma_i\sigma_i^0\sigma^0 \pi_t^* + f(\lambda_i) + b_i  \\
  & \leq K +(\gamma_i-1)[\sigma_i^2 + (\sigma_i^0)^2] + f(\lambda_i) + b_i  \\
  & \leq K + \sup_{i\geq 1}(\gamma_i-1)[\sigma_i^2+(\sigma_i^0)^2] + \sup_{i\geq 1}f(\lambda_i) + \sup_{i\geq 1}b_i := C.
\end{align*}
It holds that $(\epsilon_i)^{\gamma_i -1} \leq \frac{C}{f(\lambda_i)}\leq \frac{C}{\inf_{i\geq 1}f(\lambda_i)}$. Therefore, we obtain that
\begin{align*}
  \epsilon_i & \geq \left(\frac{\inf_{i\geq 1}\lambda_i}{C}\right)^{\frac{1}{1-\gamma_i}} \geq  \inf_{i\geq1} \left(\frac{\inf_{i\geq 1}\lambda_i}{C}\right)^{\frac{1}{1-\gamma_i}} := \tilde{\epsilon},
\end{align*}
which completes the proof.
\end{proof}

\begin{proof}[Proof of Lemma \ref{lem:auxiliarylemmas}]
(i) By \eqref{eq:pi-epsilon-i}, we have that $\pi_t^{*,i,n} = \phi_i(t,\Lambda_t^{f,i,n}) \in [D_0,1-\epsilon_0]$. Also recall that the wealth process $X^{*,i,n}=(X_t^{*,i,n})_{t\in[0,T]}$ satisfies \eqref{eq:wealth-X-i-ast}. The conclusion clearly holds when $p=0$. It suffices to prove the result when $p\neq 0$. By It\^{o} formula, we have that
\begin{align}\label{eq:X-i-star-p}
 d(X_t^{*,i,n})^p
 & = p(X_t^{*,i,n})^p(r+ (b_i +\Lambda_t^{f,i,n})\pi_t^{*,i,n}) +\frac{1}{2}p(p-1) (\pi_t^{*,i,n})^2(\sigma_i^2+(\sigma_i^0)^2)]dt \nonumber \\
 &\quad + p(X_t^{*,i,n})^p \pi_t^{*,i,n}(\sigma_i dW_t^i +\sigma_i^0dW_t^0) + (X_{t-}^{*,i,n})^p\log(1-\pi_t^{*,i,n})^p dN_t^{i,n}.
\end{align}
It follows that
\begin{align}\label{eq:X-i-star-expect}
 \Ex\left[\left(X_t^{*,i,n}\right)^p\right]& = x_i^p \Ex\Bigg[\exp\Bigg\{p\int_{0}^{t}\left[r+(b_i +\Lambda_s^{f,i,n})\pi_s^{*,i,n} - \frac{1}{2}(\pi_s^{*,i,n})^2(\sigma_i^2+(\sigma_i^0)^2)\right]ds\nonumber \\
 & \quad + p\int_{0}^{t}\pi_s^{*,i,n}(\sigma_i dW_s^i +\sigma_i^0dW_s^0)+ \int_{0}^{t}\log(1-\pi_s^{*,i,n})^p dN_s^{i,n} \Bigg\}\Bigg] \nonumber \\
 & = x_i^p \Ex^{\widehat{\mathbb{Q}}}\Bigg[\exp\Bigg\{ \int_{0}^{t} p\left[r+ b_i \pi_s^{*,i,n} + \frac{1}{2}(p-1) (\pi_s^{*,i,n})^2(\sigma_i^2+(\sigma_i^0)^2)\right]\nonumber\\
 &\quad + ((1-\pi_s^{*,i,n})^p +p\pi_s^{*,i,n}-1)\Lambda_s^{f,i,n} ds \Bigg\}\Bigg].
\end{align}
Here, $\widehat{\Qx}\sim\Px$ is defined by $\frac{d\widehat{\mathbb{Q}} }{d\mathbb{P}}\big|_{\mathcal{G}_t}=\widehat{L}_t$, where $\hat{L}=(\hat{L}_t)_{t\in[0,T]}$ satisfies SDE under $\Px$ that
\begin{equation}\label{eq:Girsanov-L-2}
  \frac{d\hat{L}_t}{\hat{L}_{t-}} = p\sigma_i \pi_t^{*,i,n}dW_t^i + p\sigma_i^0 \pi_t^{*,i,n} dW_t^0 + [(1-\pi_t^{*,i,n})^p -1]dM_t^{f,i,n},\quad\hat{L}_0=1.
\end{equation}
Note that both of $\pi^{*,i,n}=(\pi_t^{*,i,n})_{t\in[0,T]}$ and $\Lambda^{f,i,n}=(\Lambda_t^{f,i,n})_{t\in[0,T]}$ are bounded. The desired estimate~\eqref{eq:bound-X-i-ast} follows from \eqref{eq:X-i-star-expect}.

(ii) Note that the mean field intensity factor process reduces to
\begin{equation}\label{eq:lambda-MF-000}
  \lambda_t^{f,o} = f(\lambda_t^{l}),\quad \lambda_t^{l} = \lambda_0 + \int_{0}^{t}\alpha(\lambda_{\infty} -\lambda_s^{l})ds + \int_{0}^{t} \beta \varsigma f(\lambda_s^{l})ds,
\end{equation}
and the intensity process $\Lambda^{f,i,n} = (\Lambda_t^{f,i,n})_{t\in[0,T]}$ takes the form that
 \begin{equation}\label{eq:Lambdain00}
   \Lambda_t^{f,i,n} = f(\lambda_t^{i,n}),\quad d\lambda_t^{i,n} = \alpha_i(\lambda_{\infty}^i - \lambda_t^{i,n})dt + \frac{\beta_i}{n}\sum_{j=1}^{n} \varsigma_{j}dN_t^{j,n}.
\end{equation}
It follows from \eqref{eq:lambda-MF-000} and \eqref{eq:Lambdain00} that
\begin{equation*}
  d(\lambda_t^{i,n} - \lambda_t^l) = \left[\alpha_i(\lambda_{\infty}^i - \lambda_t^{i,n}) -\alpha(\lambda_{\infty} -\lambda_t^{l}) \right]dt + \frac{\beta_i}{n}\sum_{j=1}^{n}\varsigma_{j}dN_t^{j,n} - \beta \varsigma f(\lambda_t^l) dt.
\end{equation*}
Applying It\^{o}'s lemma, we obtain that
\begin{align*}
   d(\lambda_t^{i,n} - \lambda_t^l)^2  & = 2 (\lambda_t^{i,n} - \lambda_t^l)\left[ \alpha_i(\lambda_{\infty}^i - \lambda_t^{i,n})-\alpha(\lambda_{\infty} -\lambda_t^{l})  - \beta \varsigma f(\lambda_t^l) \right]dt \\
  & \quad + \sum_{j=1}^{n} \left[ \left(\lambda_t^{i,n} - \lambda_t^l + \frac{\beta_i}{n}\varsigma_{j} \right)^2 -(\lambda_t^{i,n} - \lambda_t^l)^2  \right]dN_t^{j,n}.
\end{align*}
Taking the integral from $0$ to $t$ and then taking expectations on both sides, we arrive at
\begin{align*}
  \mathbb{E}\left[ \left(\lambda_t^{i,n} - \lambda_t^l\right)^2  \right] & = \int_{0}^{t} \Ex\left[ 2 (\lambda_s^{i,n} - \lambda_s^l)\left( \alpha_i(\lambda_{\infty}^i - \lambda_s^{i,n})-\alpha(\lambda_{\infty} -\lambda_s^{l}) - \beta \varsigma f(\lambda_s^l) \right) \right] ds  + O\left(\frac{1}{n^2}\right) \\
  &\quad + \frac{1}{n} \sum_{j=1}^{n}\int_{0}^{t} 2\beta_i\varsigma_j \mathbb{E}\left[(\lambda_s^{i,n} - \lambda_s^l) f(\lambda_s^{j,n})\right]ds \\
  & = \int_{0}^{t} \mathbb{E}\left[ 2 (\lambda_s^{i,n} - \lambda_s^l)\left( \alpha_i\big(\lambda_{\infty}^i - \lambda_s^{i,n}\big)-\alpha(\lambda_{\infty} -\lambda_s^{l}) \right) \right] ds  + O\left(\frac{1}{n^2}\right) \\
  &\quad + \frac{1}{n} \sum_{j=1}^{n}\int_{0}^{t} 2 \mathbb{E}\left[(\lambda_s^{i,n} - \lambda_s^l)\left( \beta_i\varsigma_j  f(\lambda_s^{j,n}) - \beta \varsigma f(\lambda_s^l)\right)\right]ds.
\end{align*}
It follows from the Lipschitz continuity of $f$ that, there exists a constant $C>0$ independent of $n$ such that
\begin{align*}
 & \frac{1}{n}\sum_{i=1}^{n} \mathbb{E}\left[ \left(\lambda_t^{i,n} - \lambda_t^l\right)^2  \right] \\
  & = \frac{1}{n}\sum_{i=1}^{n} \int_{0}^{t} \mathbb{E}\left[ 2 (\lambda_s^{i,n} - \lambda_s^l)\left( \alpha_i\big(\lambda_{\infty}^i - \lambda_s^{i,n}\big)-\alpha(\lambda_{\infty} -\lambda_s^{l}) \right) \right] ds  + O\left(\frac{1}{n^2}\right) \\
  &\quad + \frac{1}{n^2}\sum_{i=1}^{n}\sum_{j=1}^{n} \int_{0}^{t} 2 \mathbb{E}\left[(\lambda_s^{i,n} - \lambda_s^l)\left( \beta_i\varsigma_j  f(\lambda_s^{j,n}) - \beta \varsigma f(\lambda_s^l)\right)\right]ds \\
 &\quad\leq \frac{1}{n}\sum_{i=1}^{n}  \int_{0}^{t} C \Ex\left[ (\lambda_s^{i,n} - \lambda_s^l)^2 \right] ds  + O\left(\frac{1}{n^2}\right)
 + \frac{1}{n} \sum_{i=1}^{n}\int_{0}^{t}  C \Ex\left[(\lambda_s^{i,n} - \lambda_s^l)\frac{1}{n} \sum_{j=1}^{n}(\lambda_s^{j,n} -\lambda_s^l) \right]ds \\
  &\quad \leq C \int_{0}^{t} \frac{1}{n}\sum_{i=1}^{n} \Ex\left[\left(\lambda_s^{i,n} - \lambda_s^l\right)^2 \right] ds  + O\left(\frac{1}{n^2}\right).
\end{align*}
By using Gronwall's inequality, we conclude that
\begin{equation*}
  \frac{1}{n}\sum_{i=1}^{n} \Ex\left[ \left(\lambda_t^{i,n} -\lambda_t^l\right)^2 \right] = O\left(\frac{1}{n^2}\right).
\end{equation*}
Applying Jensen inequality and the Lipschitz continuity property, it holds that
\begin{align*}
 \Ex\left[\left(\bar{\Lambda}_t^{f,n} - \lambda_t^{f,o}\right)^2  \right] & = \Ex\left[\left|\frac{1}{n}\sum_{i=1}^{n}\left( f(\lambda_t^{i,n}) - f(\lambda_t^l)\right) \right|^2  \right]\leq \frac{1}{n}\sum_{i=1}^{n} \Ex\left[\left(f(\lambda_t^{i,n}) - f(\lambda_t^l)\right)^2  \right] \\
 &\leq C\frac{1}{n}\sum_{i=1}^{n} \Ex\left[\left(\lambda_t^{i,n} - \lambda_t^l \right)^2  \right]= O\left(\frac{1}{n^2}\right).
\end{align*}
The desired claim \eqref{eq:lem-lambda} then holds.

(iii)
Recall that $\bar{X}_t^{*,n}=(\prod_{i=1}^n X_t^{*,i,n})^{\frac{1}{n}}$ where $ X_t^{*,i,n}$ is defined by  \eqref{eq:wealth-X-i-ast}, and the geometric mean process $m^*= (m_t^{*})_{t\in[0,T]}$ satisfies \eqref{eq:geometric-m-thm}. Let us define $Y_t^i:= \log(X_t^{*,i,n}) $ and
\begin{equation*}
  \bar{Y}_t := \frac{1}{n}\sum_{i = 1}^{n}Y_t^i = \log(\bar{X}_t^{*,n}) = \frac{1}{n} \sum_{i=1}^{n}\log(X_t^{*,i,n}).
\end{equation*}
It follows from \eqref{eq:wealth-X-i-ast} that
\begin{align}\label{eq:log-Y-ast-i}
  dY_t^i &= \left[r+ (b_i+\Lambda_t^{f,i,n})\pi_t^{*,i,n} - \frac{1}{2}\left(\sigma_i^2(\pi_t^{*,i,n})^2+ (\sigma_i^0 \pi_t^{*,i,n})^2\right)\right]dt + \pi_t^{*,i,n}\sigma_idW_t^i + \pi_t^{*,i,n}\sigma_i^0dW_t^0 \nonumber \\
  &\quad + \log(1-\pi_{t}^{*,i,n})dN_t^{i,n}.
\end{align}
We denote $ Z_t :=\log(m_t^*) $, which satisfies that
\begin{equation}\label{eq:Z-t}
  dZ_t =  \left\{r+ (b+\lambda_t^{f,o})\pi_t^{\ast} - \frac{1}{2}(\sigma^2+\sigma_0^2) (\pi_t^{\ast})^2
 + \lambda_t^o\log(1-\pi_t^{\ast})\right\}dt   + \sigma^0\pi_t^{\ast}dW_t^0.
\end{equation}
Note that
\begin{align*}
 & \Ex\left[\left| \bar{X}_t^{*,n} - m_t^* \right|^2 \right]= \Ex\left[ e^{2Z_t} \left( e^{\bar{Y}_t -Z_t} -1 \right)^2\right]\\
 &\qquad\leq \Ex\left[ e^{4 Z_t} \right]^{\frac{1}{2}} \Ex\left[\left( e^{2(\bar{Y}_t -Z_t) } - 2e^{\bar{Y}_t -Z_t } + 1 \right)^2 \right]^{\frac{1}{2}}\\
 &\qquad =  \Ex\left[ e^{4 Z_t} \right]^{\frac{1}{2}} \Ex\left[ e^{4(\bar{Y}_t -Z_t)} + 4e^{2(\bar{Y}_t -Z_t) } + 1 - 4e^{3(\bar{Y}_t -Z_t) } + 2e^{2(\bar{Y}_t -Z_t) } - 4e^{(\bar{Y}_t -Z_t)}\right]^{\frac{1}{2}}.
\end{align*}
To prove the claim \eqref{eq:lem-bar-X-m},  by the boundedness of $\bar{X}^{*,n}=(\bar{X}_t^{*,n})_{t\in[0,T]}$, it is sufficient to prove that, for any $t\in[0,T]$
\begin{equation}\label{eq:bar-Y-Y}
   \lim_{n\rightarrow \infty} \Ex\left[ e^{\bar{Y}_t -Z_t} \right] = 1.
\end{equation}
To this end, for $i = 1,\ldots,n$, we introduce the auxiliary SDE that
\begin{align*}
 d \hat{Y}_t^i &= \left[r+ (b_i+\Lambda_t^{f,i,n} )\pi_t^{*,i,n}- \frac{1}{2}\left(\sigma_i^2+ (\sigma_i^0)^2\right) (\pi_t^{*,i,n})^2 + \Lambda_t^{f,i,n}\log(1-\pi_t^{*,i,n}) \right]dt + \pi_t^{*,i,n}\sigma_i^0dW_t^0,
\end{align*}
and
\begin{align*}
  d\underline{Y}_t^i&= \left[r+ (b+\Lambda_t^{f,i,n})\pi_t^{*,i,n}- \frac{1}{2}\left(\sigma^2+ (\sigma^0)^2 \right)(\pi_t^{*,i,n})^2 +  \Lambda_t^{f,i,n} \log(1-\pi_t^{*,i,n}) \right]dt + \pi_t^{*,i,n}\sigma^0 dW_t^0.
\end{align*}
For some positive constants $p_i$ for $i=1,2,3$ satisfying $\sum_{i=1}^3 \frac{1}{p_i}=1$, we can derive by generalized H\"{o}lder inequality that
\begin{align}\label{holdineql}
  \Ex\left[ e^{\bar{Y}_t -Z_t} \right]
  & =  \Ex\left[  e^{(\bar{Y}_t - \frac{1}{n}\sum_{i=1}^{n}\hat{Y}_t^i) +(\frac{1}{n}\sum_{i=1}^{n} \hat{Y}_t^i - \frac{1}{n}\sum_{i=1}^{n} \underline{Y}_t^i)
  + (\frac{1}{n}\sum_{i=1}^{n}\underline{Y}_t^i -Z_t)}\right]\\
  & \leq  \Ex\left[ e^{ p_1 (\bar{Y}_t - \frac{1}{n}\sum_{i=1}^{n}\hat{Y}_t^i) }\right]^{1/p_1} \Ex\left[e^{ p_2(\frac{1}{n}\sum_{i=1}^{n}\hat{Y}_t^i - \frac{1}{n}\sum_{i=1}^{n}\underline{Y}_t^i) }\right]^{1/p_2}
  \Ex\left[e^{ p_3( \frac{1}{n}\sum_{i=1}^{n}\underline{Y}_t^i -Z_t) }\right]^{1/p_3}. \notag
\end{align}
For the first term on the RHS of \eqref{holdineql}, we have that
\begin{align*}
  \Ex\left[ e^{\frac{1}{n}\sum_{i=1}^{n}(Y_t^i -\hat{Y}_t^i)}\right] & = \Ex\left[e^{\frac{1}{n}\sum_{i=1}^{n} \int_{0}^{t}\pi_s^{*,i,n}\sigma_idW_s^i + \frac{1}{n}\sum_{i=1}^{n} \int_{0}^{t}\log\left(1-\pi_{s}^{*,i,n}\right)dM_s^{f,i}}\right]\\
  & = 1 + O\left(\frac{1}{n}\right) .
\end{align*}
For the second term on the RHS of \eqref{holdineql}, we have that
\begin{align*}
 \Ex\left[e^{\frac{1}{n}\sum_{i=1}^{n}(\hat{Y}_t^i - \underline{Y}_t^i)}\right]
 & = \Ex\Bigg[\exp\Bigg\{\frac{1}{n}\sum_{i=1}^{n}\int_{0}^t\left[(b_i-b)\pi_s^{*,i,n}- \frac{1}{2}\left(\sigma_i^2+ (\sigma_i^0)^2 - \sigma^2- (\sigma^0)^2\right)(\pi_s^{*,i,n})^2 \right] ds \\
 & \qquad + \frac{1}{n}\sum_{i=1}^{n}\int_{0}^{t}(\sigma_i^0-\sigma^0)\pi_s^{*,i,n} dW_s^0 \Bigg\} \Bigg]\\
 & \leq \Ex\left[\exp\left\{\frac{1}{n}\sum_{i=1}^{n} C \int_{0}^T\left[|b_i-b| + |\sigma_i^2- \sigma^2|+ |(\sigma_i^0)^2 - (\sigma^0)^2|\right] ds\right\}\right],
\end{align*}
where $C$ is a positive constant independent of $n$. By the assumption $\bm{(A_O)}$, we have that
\begin{align*}
    \Ex\left[e^{ \frac{1}{n}\sum_{i=1}^{n}( \hat{Y}_t^i - \underline{Y}_t^i)} \right]
   = 1+ O\left(\frac{1}{n}\right).
\end{align*}
For the third term on the RHS of \eqref{holdineql}, it holds that
\begin{align*}
 & \Ex\left[ e^{\frac{1}{n}\sum_{i=1}^{n}(\underline{Y}_t^i -Z_t)}\right]
 = \Ex\Bigg[ \exp\Bigg\{\frac{1}{n}\sum_{i=1}^{n}\int_{0}^{t}\big[(\Lambda_s^{f,i,n}\pi_s^{*,i,n}- \lambda_s^{f,o}\pi_s^{*}) - \frac{1}{2}(\sigma^2+ (\sigma^0)^2) [(\pi_s^{*,i,n})^2-(\pi_s^{*})^2] \\
 &\qquad\quad +  \Lambda_s^{f,i,n}\log(1-\pi_s^{*,i,n}) - \lambda_s^{f,o}\log(1-\pi_s^{*})\big]ds + \frac{1}{n}\sum_{i=1}^{n} \int_{0}^{t}\sigma^0(\pi_s^{*,i,n} - \pi_s^{*})dW_s^0 \Bigg\} \Bigg]\\
 &\quad\leq \Ex\left[\exp \left\{ C \int_{0}^{t}\frac{1}{n}\sum_{i=1}^{n} \left[(\pi_s^{*,i,n}- \pi_s^{*}) + \log(1-\pi_s^{*,i,n}) - \log(1-\pi_s^{*})\right]ds \right\}\right].
\end{align*}
We then claim that
\begin{align}\label{lastclaim}
\Ex\left[ e^{ \frac{1}{n}\sum_{i=1}^{n}(\underline{Y}_t^i -Z_t)}\right] = 1+ O\left(\frac{1}{n}\right).
\end{align}
Recall that $\pi_t^{*,i,n}:= \phi_i(t,\Lambda_t^{f,i,n})$ in \eqref{eq:pi-epsilon-i} satisfies the equation that
\begin{equation*}
 (\gamma_i-1)[\sigma_i^2+(\sigma_i^0)^2] \pi_t^{*,i,n} - \theta_i\gamma_i\sigma_i^0\sigma^0\pi_t^{*}  - \Lambda_t^{f,i,n}(1-\pi_t^{*,i,n})^{\gamma_i-1} +\Lambda_t^{f,i,n} + b_i = 0,
\end{equation*}
and $\pi_t^{*} = \phi(\lambda_t^{f,o})$ given by \autoref{thm:optimal-MFG} is the solution to
\begin{equation*}
  (\gamma-1)[ \sigma^2 + (\sigma^0)^2] \pi_t^*   - \theta\gamma(\sigma^0)^2\pi_t^*  - \lambda_t^{f,o} (1-\pi_t^*)^{\gamma -1} +\lambda_t^{f,o} + b = 0.
\end{equation*}
We introduce an auxiliary control $ \hat{\pi}_t^i := \phi_i(t,\lambda_t^{f,o}) $, which satisfies
\begin{equation*}
  (\gamma_i-1)[\sigma_i^2 + (\sigma_i^0)^2]\hat{\pi}_t^i - \theta_i\gamma_i\sigma_i^0\sigma^0\pi_t^*  - \lambda_t^{f,o} (1-\hat{\pi}_t^i)^{\gamma_i -1} +\lambda_t^{f,o} + b_i = 0.
\end{equation*}
From the proof of Lemma \ref{lem:auxiliarylemmas}-(ii) and the assumption $\bm{(A_O)}$, we have that $\hat{\pi}_t^i \rightarrow  \pi_t^{*} $ as $ i \rightarrow \infty$, and there exists a constant $C>0$ independent of $n$ such that
\begin{align*}
 \frac{1}{n}\sum_{i=1}^{n} (\pi_t^{*,i,n} - \pi_t^{\ast})
 & = \frac{1}{n}\sum_{i=1}^{n}(\pi_t^{*,i,n} - \hat{\pi}_t^{i}) + \frac{1}{n}\sum_{i=1}^{n} (\hat{\pi}_t^i - \pi_t^{*})\\
 & = \frac{1}{n}\sum_{i=1}^{n}(\phi_i(t,\Lambda_t^{f,i,n}) - \phi_i(t,\lambda_t^{f,o}))
     + \frac{1}{n}\sum_{i=1}^{n}( \phi_i(t,\lambda_t^{f,o}) - \phi(\lambda_t^{f,o}) )\\
 & \leq \frac{1}{n}\sum_{i=1}^{n} C \left|\Lambda_t^{f,i,n} - \lambda_t^{f,o} \right| + O\left(\frac{1}{n}\right).
\end{align*}
Moreover, the order of the term $\frac{1}{n}\sum_{i=1}^{n}[\log(1-\pi_t^{*,i,n}) - \log(1-\pi_t^{\ast})]$ is the same to the order of $\frac{1}{n}\sum_{i=1}^{n}(\pi_t^{*,i,n} - \pi_t^{*})$ because the function $\log(1-x) $ is Lipschitz continuous in $x\in[D_0,1-\epsilon_0]$, which proves that the claim \eqref{lastclaim} holds. Putting all the pieces together completes the proof.
\end{proof}

\section{Arguments to Derive $\lambda^{f,o}$ in the Mean Field Model}\label{sec:heu}

Let us first recall that the default intensity process for agent $i$ satisfies SDE~\eqref{eq:lambda-i}, for $i=1,\ldots,n$,
\begin{align}\label{eq:Lambda-i-appx-2}
  \Lambda_t^{f,i} = f(\lambda_t^i),  \ \ \ d\lambda_t^i =\alpha_i(\lambda_{\infty}^i -\lambda_t^i)dt + \frac{\beta_i}{n}\sum_{j=1}^{n}\varsigma_jdN^j_t,
\end{align}
where $M^{f,i}_t := N^i_t - \int_{0}^{t}\Lambda_s^{f,i} ds$, $t\in[0,T]$, is a $(\Px,\Gx)$-martingale. 
For $i=1,\ldots,n$, recall that the type vector of the default intensity model is $o^i = (\lambda_0^i,\alpha_i,\lambda^i_\infty, \beta_i, \varsigma_i)\in \mathcal{O}:= \R_{+}^5$, and the space $E := \mathcal{O}\times\R_{+}$. We then define the following empirical measure-valued process on ${\cal B}(E)$ by
\begin{align}\label{eq:mu-n-appx}
  \mu_t^n := \frac{1}{n}\sum_{i=1}^{n}\delta_{(o^i,\Lambda_t^{f,i})},\quad \forall t\in[0,T].
\end{align}
We next claim that, under the assumption $\bm{(A_O)}$, the mean field limit of the default intensity process satisfies that
\begin{align}\label{eq:lambda-MFG-appx}
   \lambda_t^{f,o} = f(\lambda_t^l),\quad  \lambda_t^l = \lambda_0 + \int_{0}^{t}\alpha(\lambda_{\infty} -\lambda_s^l)ds + \int_{0}^{t} \beta\varsigma f(\lambda_s^l)ds,
\end{align}
where $o=(\lambda_0,\alpha,\lambda_\infty, \beta, \varsigma)\in{\cal O}$ is the constant vector in the assumption $\bm{(A_O)}$ for the mean field model.
In fact, by applying It\^{o}'s formula to an arbitrary $\varphi\in C_c^2(E)$ (a test function that is $C_c^2$ w.r.t. $\lambda\in\R_+$), we obtain
\begin{align}\label{eq:f-Lambda-H-X-i}
\varphi(\Lambda_t^{f,i})& =\varphi(\Lambda_0^{f,i})+\int_{0}^{t} \varphi'(\Lambda_s^{f,i})f'(\lambda_s^i)\alpha_i(\lambda_{\infty}^i -\lambda_s^i)ds
+\sum_{j=1}^{n}\int_{0}^{t}\{\varphi(f(\lambda_{s-}^{i} +\frac{\beta_{ij}}{n})) - \varphi(\Lambda_{s-}^{f,i})\}\Lambda_s^{f,j}ds  \nonumber \\
& \quad  + \sum_{j=1}^{n}\int_{0}^{t}\{\varphi(f(\lambda_{s-}^{i} +\frac{\beta_{ij}}{n})) - \varphi(\Lambda_{s-}^{f,i})\}dM_s^{f,j},
\end{align}
where $\beta_{ij}:=\beta_i\varsigma_j$ for $i,j=1,\ldots,n$. It follows from mean value theorem that $\varphi(f(\lambda_{s-}^{i}+\frac{\beta_{ij}}{n})) - \varphi(\Lambda_{s-}^{f,i})=\frac{\beta_{ij}}{n}\int_0^1 f'(\lambda_{s-}^{i} +\frac{\theta}{n}\beta_{ij})\varphi'(f(\lambda_{s-}^{i} +\frac{\theta}{n}\beta_{ij}))d\theta$. Thus, we arrive at
\begin{align}\label{eq:f-Lambda-H-X-iaverage}
\frac{1}{n}\sum_{i=1}^n\varphi(\Lambda_t^{f,i})& = \frac{1}{n}\sum_{i=1}^n\varphi(\Lambda_0^{f,i})+\int_{0}^{t} \frac{1}{n}\sum_{i=1}^n \varphi'(f(\lambda_s^{i}))f'(\lambda_s^i)\alpha_i(\lambda_{\infty}^i -\lambda_s^i)ds\nonumber\\
&\quad +\int_0^1 \int_{0}^{t}\frac{1}{n}\sum_{i=1}^n\left(\sum_{j=1}^{n}f'(\lambda_{s-}^{i} +\frac{\theta}{n}\beta_{ij})\varphi'(f(\lambda_{s-}^{i} +\frac{\theta}{n}\beta_{ij})) \frac{1}{n}\beta_{ij}\Lambda_s^{f,j}\right)dsd\theta  \nonumber \\
& \quad  + \frac{1}{n}\sum_{i=1}^n\int_{0}^{t}\{\varphi(f(\lambda_{s-}^{i} +\frac{\beta_{ij}}{n})) - \varphi(\Lambda_{s-}^{f,i})\}dM_s^{f,j}.
\end{align}
We define $\langle\varphi,\mu\rangle:=\int_E\varphi(e)\mu(de)$ for any $\mu\in{\cal P}(E)$. Then, Eq.~\eqref{eq:f-Lambda-H-X-iaverage} can be read as:
\begin{align}\label{eq:f-Lambda-H-X-iaverage2}
\langle\varphi,\mu_t^n\rangle& = \langle\varphi,\mu_0^n\rangle+\langle{\cal L}^1\varphi,\mu_t^n\rangle+\frac{1}{n}\sum_{i=1}^n\int_{0}^{t}\{\varphi(f(\lambda_{s-}^{i} +\frac{\beta_{ij}}{n})) - \varphi(\Lambda_{s-}^{f,i})\}dM_s^{f,j}\nonumber\\
&\quad +\int_0^1 \int_{0}^{t}\frac{1}{n}\sum_{i=1}^n\left(\sum_{j=1}^{n}f'(\lambda_{s-}^{i} +\frac{\theta}{n}\beta_{ij})\varphi'(f(\lambda_{s-}^{i} +\frac{\theta}{n}\beta_{ij})) \frac{1}{n}\beta_{ij}\Lambda_s^{f,j}\right)dsd\theta.
\end{align}
where we defined the operators ${\cal L}^1\varphi(e):= \varphi'(f(\lambda)) f'(\lambda)\alpha(\lambda_{\infty} -\lambda)$ for $e\in E$, ${\cal L}^2\varphi(e):=\beta f'(\lambda)\varphi'(f(\lambda))$ for $e\in E$, and $q(e):=\varsigma f(\lambda)$ for $e\in E$. Let $\mu=(\mu_t)_{t\in[0,T]}$ be the weak limit point of $\mu^n=(\mu^n_t)_{t\in[0,T]}$ as $n\to\infty$. Then, the 4th term of RHS of Eq.~\eqref{eq:f-Lambda-H-X-iaverage2} converges to $\int_0^t\langle{\cal L}^2\varphi,\mu_s\rangle \langle q,\mu_s\rangle ds$ as $n\to\infty$, a.s. On the other hand, as $M^{f,j}=(M_t^{f,j})_{t\in[0,T]}$ is a martingale, it follows from the martingale limit theorem that the martingale sequence $\frac{1}{n}\sum_{i=1}^n\int_{0}^{t}\{\varphi(f(\lambda_{s-}^i+\frac{1}{n}\beta_{ij})) - \varphi(\Lambda_{s-}^i)\}dM_s^{f,j}$ converges to zero, as $n\to\infty$, a.s.. Hence, as $n\to\infty$, we have from Eq.~\eqref{eq:f-Lambda-H-X-iaverage2} that the limit point $\mu$ satisfies
\begin{align}\label{eq:f-Lambda-H-X-iaveragelimit}
\langle\varphi,\mu_t\rangle& = \langle\varphi,\mu_0\rangle+\langle{\cal L}^1\varphi,\mu_t\rangle+\int_0^t\langle{\cal L}^2\varphi,\mu_s\rangle \langle q,\mu_s\rangle ds.
\end{align}
We next claim that $\mu_t=\delta_{(o,\lambda_t^{f,o})}$ for $t\in[0,T]$. In fact, it follow from It\^o's formula that
\begin{align}\label{eq:varphilambdato222}
\varphi(\lambda_t^{f,o})&= \varphi(\lambda_0^{f,o}) + \int_0^t \varphi'(f(\lambda_s^l))f'(\lambda_s^l)\{\alpha(\lambda_{\infty} -\lambda_s^{l})+\beta\varsigma f(\lambda_s^{l}) \} ds.
\end{align}
Note that, if $\mu_t=\delta_{(o,\lambda_t^{f,o})}$, then $\langle\varphi,\mu_t\rangle=\varphi(\lambda_t^{f,o})$ and $\langle\varphi,\mu_0\rangle=\varphi(f(\lambda_0))$. Moreover, it holds that
\begin{align*}
\varphi'(f(\lambda_s^l))f'(\lambda_s^l)\{\alpha(\lambda_{\infty} -\lambda_s^{l})+\beta\varsigma f(\lambda_s^{l}) \}
&=\varphi'(f(\lambda_s^l))f'(\lambda_s^l)\alpha(\lambda_{\infty} -\lambda_s^{l}) + \varphi'(f(\lambda_s^l))f'(\lambda_s^l) \beta\varsigma f(\lambda_s^{l}) \nonumber\\
&=\langle{\cal L}^1\varphi,\mu_s\rangle + \langle{\cal L}^2\varphi,\mu_s\rangle \langle q,\mu_s\rangle.
\end{align*}
Plugging the above equality into \eqref{eq:varphilambdato222}, we have that $\mu_t=\delta_{(o,\lambda_t^{f,o})}$ for $t\in[0,T]$ indeed satisfies Eq.~\eqref{eq:f-Lambda-H-X-iaveragelimit}. Moreover, it follows from the uniqueness of Eq.~\eqref{eq:f-Lambda-H-X-iaveragelimit} (c.f. \cite{Bosifin2015} by using martingale problem) that the limit point of $\mu^n$ is given by $\mu_t=\delta_{(o,\lambda_t^{f,o})}$ for $t\in[0,T]$. This yields that the mean field limit of $(\Lambda^{f,i})_{i=1}^n$ described as \eqref{eq:Lambda-i-appx-2} is given by $\lambda^{f,o}=(\lambda_t^{f,o})_{t\in[0,T]}$ in Eq.~\eqref{eq:lambda-MFG-appx}. \hfill$\Box$

\end{appendix}

\end{document}